\newif\ifaistats
\newif\ifcolt

\newif\iffull 

\colttrue
\coltfalse 

\aistatsfalse 

\ifcolt
        \aistatsfalse 
        \fullfalse 

        \documentclass[final,12pt]{Styles_camera_ready/colt2021} 

        \usepackage{paper-ak-colt21} 
\else
        \fulltrue
		\documentclass[11pt]{article}
		\usepackage[margin = 1.0 in]{geometry}
		\usepackage{paper-ak}
\fi

\DeclareMathOperator{\nnz}{\mathsf{nnz}}
\DeclareMathOperator{\ns}{\mathsf{ns}}
\DeclareMathOperator{\sr}{\mathsf{sr}}
\DeclareMathOperator{\var}{var}

\DeclareMathOperator{\tail}{tail}
\DeclareMathOperator{\polylog}{polylog}
\DeclareMathOperator{\rsp}{\mathsf{rsp}}
\DeclareMathOperator{\csp}{\mathsf{csp}}

\def\compactify{\itemsep=0pt \topsep=0pt \partopsep=0pt \parsep=0pt}
\newcommand*\widefbox[1]{\fbox{\hspace{1em}#1\hspace{1em}}}

\ifcolt
    \title[Near-Optimal Entrywise Sampling of Numerically Sparse Matrices]{Near-Optimal Entrywise Sampling of Numerically Sparse Matrices%
    \iffull \else
    \protect\thanks{A full version is available at \href{https://arxiv.org/abs/2011.01777}{arXiv:2011.01777}.} \fi}
    
    \coltauthor{%
         \Name{Vladimir Braverman}\thanks{This research was supported in part by NSF CAREER grant 1652257, NSF grant 1934979, ONR Award N00014-18-1-2364 and the Lifelong Learning Machines program from DARPA/MTO.} \Email{vova@cs.jhu.edu}\\
         \addr Johns Hopkins University
         \AND
         \Name{Robert Krauthgamer}\thanks{Work partially supported by ONR Award N00014-18-1-2364,
  the Israel Science Foundation grant \#1086/18,
  and a Minerva Foundation grant.
} \Email{robert.krauthgamer@weizmann.ac.il}\\
         \addr Weizmann Institute of Science%
         \AND
         \Name{Aditya Krishnan}\thanks{This research was supported in part by NSF CAREER grant 1652257, ONR Award N00014-18-1-2364 and the Lifelong Learning Machines program from DARPA/MTO.} \Email{akrish23@jhu.edu}\\
         \addr Johns Hopkins University
         \AND
         \Name{Shay Sapir}\Email{shay.sapir@weizmann.ac.il}\\
         \addr Weizmann Institute of Science%
        }

\else
    \title{Near-Optimal Entrywise Sampling of Numerically Sparse Matrices}
    
    		\author{
			Vladimir Braverman%
			\thanks{This research was supported in part by NSF CAREER grant 1652257, NSF grant 1934979, ONR Award N00014-18-1-2364 and the Lifelong Learning Machines program from DARPA/MTO.} \\ Johns Hopkins University\\ \texttt{vova@cs.jhu.edu}
				\and
			Robert Krauthgamer%
			\thanks{Work partially supported by ONR Award N00014-18-1-2364,
  the Israel Science Foundation grant \#1086/18,
  and a Minerva Foundation grant.
}
			\\ Weizmann Institute of Science \\ \texttt{robert.krauthgamer@weizmann.ac.il}
				\and 
			Aditya Krishnan%
			\thanks{This research was supported in part by NSF CAREER grant 1652257, ONR Award N00014-18-1-2364 and the Lifelong Learning Machines program from DARPA/MTO.} \\ Johns Hopkins University \\ \texttt{akrish23@jhu.edu}
				\and 
			Shay Sapir \\ Weizmann Institute of Science \\ \texttt{shay.sapir@weizmann.ac.il}
		}
\fi
\usepackage{times}

\begin{document}

\maketitle

\begin{abstract}
Many real-world data sets are sparse or almost sparse. 
One method to measure this for a matrix $A\in \R^{n\times n}$ 
is the \emph{numerical sparsity}, denoted $\mathsf{ns}(A)$, 
defined as the minimum $k\geq 1$ such that 
$\|a\|_1/\|a\|_2 \leq \sqrt{k}$
for every row and every column $a$ of $A$. 
This measure of $a$ is smooth and is clearly only smaller
than the number of non-zeros in the row/column $a$.

The seminal work of \cite{achlioptas2007fast} 
has put forward the question of approximating an input matrix $A$ by entrywise sampling. 
More precisely, the goal is to quickly compute a sparse matrix $\tilde{A}$
satisfying $\|A - \tilde{A}\|_2 \leq \epsilon \|A\|_2$
(i.e., additive spectral approximation) given an error parameter $\epsilon>0$.
The known schemes sample and rescale a small fraction of entries from $A$.

We propose a scheme that sparsifies an almost-sparse matrix $A$ ---
it produces a matrix $\tilde{A}$ with $O(\epsilon^{-2}\mathsf{ns}(A) \cdot n\ln n)$ non-zero entries with high probability.
We also prove that this upper bound on $\mathsf{nnz}(\tilde{A})$ 
is \emph{tight} up to logarithmic factors. 
Moreover, our upper bound improves when the spectrum of $A$ decays quickly
(roughly replacing $n$ with the stable rank of $A$). 
Our scheme can be implemented in time $O(\mathsf{nnz}(A))$ when $\|A\|_2$ is given.
Previously, a similar upper bound was obtained by~\cite{achlioptas2013near} but only for a restricted class of inputs that does not even include symmetric or covariance matrices. 
Finally, we demonstrate two applications of these sampling techniques, 
to faster approximate matrix multiplication,
and to ridge regression by using sparse preconditioners. 
\end{abstract}

\pagenumbering{arabic}

\section{\ifaistats INTRODUCTION \else Introduction \fi}
Matrices for various tasks in machine learning and data science often contain millions or even billions of dimensions. At the same time, they often possess structure that can be exploited to design more efficient algorithms. Sparsity in the rows and/or columns of the matrix is one such phenomenon for which many computational tasks on matrices admit faster algorithms, e.g., low-rank approximation~\citep{ghashami2016efficient,pmlr-v80-huang18a}, 
regression problems~\citep{johnson2013accelerating} and semi-definite programming~\citep{d2011semidefsubsampling,arora2005fast}. Sparsity, however, is not a numerically smooth quantity. Specifically, for a vector $x \in \R^n$ 
to be $k$-sparse, at least $n-k$ entries of $x$ must be $0$.
In practice, many entries could be small but non-zero, e.g.~due to noise, and thus the vector would be considered dense.  


A smooth analogue of sparsity for a matrix $A \in \R^{m \times n}$
can be defined as follows. 
First, for a row (or column) vector $a \in \R^n$, define its \emph{numerical sparsity} \citep{lopes2013estimating,gupta2018exploiting} to be 
\begin{equation} \label{eq:ns}
  \ns(a) \eqdef 
  \min\{ k\geq 0: \|a\|_1 \leq \sqrt{k}\|a\|_2 \} .
\end{equation}
This value is clearly at most the number of non-zeros in $a$, denoted $\|a\|_0$, but can be much smaller. 
Earlier work used variants of this quantity,
referring to $\ns(a)$ as the $\ell_1/\ell_2$-sparsity of the vector
\citep{hoyer2004non,hurley2009comparing}.
We further define the numerical sparsity of a matrix $A$, denoted $\ns(A)$, 
to be the maximum numerical sparsity of any of its rows and columns. 


In order to take advantage of sparse matrices in various computational tasks, a natural goal is to approximate a matrix $A$ with numerical sparsity $\ns(A)$ with another matrix $\tilde{A}$ of the same dimensions, that is $k$-sparse for $k=O(\ns(A))$ (i.e., every row and column is $k$-sparse). 
The seminal work of ~\cite{achlioptas2007fast} introduced a framework for matrix sparsification via entrywise sampling for approximating the matrix $A$ in spectral-norm. Specifically, they compute a sparse matrix $\tilde{A}$ by sampling and rescaling a small fraction of entries from $A$ such that with high probability $\|A - \tilde{A}\|_2 \leq \epsilon\|A\|_2$ for some error parameter $\epsilon > 0$, where $\|\cdot\|_2$ denotes the spectral-norm. 
This motivates the following definition. 

\begin{definition}
An \emph{$\epsilon$-spectral-norm approximation} for $A \in \R^{m \times n}$ is a matrix $\tilde{A} \in \R^{m \times n}$ satisfying 
\begin{equation}  \label{eq:epsapprox}
  \|\tilde{A} - A\|_2 \leq \epsilon\|A\|_2. 
\end{equation}
When $\tilde{A}$ is obtained by sampling and rescaling entries from $A$, we call it an \emph{$\epsilon$-spectral-norm sparsifier}.
\end{definition}

Before we continue, let us introduce necessary notations. Here and throughout, 
we denote the number of non-zero entries in a matrix $A$ by $\nnz(A)$, the Frobenius-norm of $A$ by $\|A\|_F$,
the stable-rank of $A$ by $\sr(A) \eqdef \|A\|_F^2/\|A\|_2^2$,
the $i$-th row and the $j$-th column of $A$ by $A_i$ and $A^j$, respectively,
and the row-sparsity and column-sparsity of $A$ by $\rsp(A)\eqdef \max_i \|A_i\|_0$ and 
$\csp(A)\eqdef \max_j \|A^j\|_0$, respectively. 

The framework of~\cite{achlioptas2007fast} can be used as a preprocessing step that ``sparsifies" numerically sparse matrices in order to speed up downstream tasks. It thus motivated a line of  work on sampling schemes \citep{arora2006fast,gittens2009error,drineas2011note,nguyen2015tensor,achlioptas2013near,kundu2014note,kundu2017recovering}, in which the output $\tilde{A}$ is an unbiased estimator of $A$, and the sampling distributions are simple functions of $A$ and hence can be computed easily, say, in nearly $O(\nnz(A))$-time and with one or two passes over the matrix. Under these constraints, the goal is simply to minimize the sparsity of the $\epsilon$-spectral-norm sparsifier $\tilde{A}$. 

The latest work, by \cite{achlioptas2013near}, provides a bound for a restricted class 
of ``data matrices". Specifically, they look at matrices $A \in \R^{m \times n}$ 
such that $\min_i \|A_i\|_1\geq \max_j \|A^j\|_1$, which can be a reasonable assumption 
when $m \ll n$.
This restricted class does not include the class of square matrices, 
and hence does not include symmetric matrices such as covariance matrices.
Hence, an important 
question is whether their results extend to a larger class of matrices. 
Our main result, described in the next section, resolves this concern in the affirmative.

\subsection{Main Results}

We generalize the sparsity bound of~\cite{achlioptas2013near}, which is the best currently known, to all matrices $A \in \R^{m \times n}$. 
Our main result is a sampling scheme to compute an $\epsilon$-spectral-norm sparsifier 
for numerically sparse matrices $A$, as follows.

\begin{restatable}{theorem}{sparsifier}\label{thm:spectralSparsifier}
There is an algorithm that, given a matrix $A\in\R^{m\times n}$ and a parameter $\epsilon>0$, where $m\geq n$,
computes with high probability 
an $\epsilon$-spectral-norm sparsifier $\tilde{A}$ for $A$ 
with expected sparsity 
\ifaistats 
$\E(\nnz(\tilde{A})) = O\left(\epsilon^{-2}\ns(A) \sr(A)\log m + \epsilon^{-1}\sqrt{\ns(A)\sr(A) n}\log m \right).$ 
\else 
$$\E(\nnz(\tilde{A})) = O\left(\epsilon^{-2}\ns(A) \sr(A)\log m + \epsilon^{-1}\sqrt{\ns(A)\sr(A) n} \log m \right).$$ 
\fi
Moreover, it runs in $O(\nnz(A))$-time when a constant factor estimate of $\|A\|_2$ is given.%
\footnote{A constant factor estimate of $\|A\|_2$ can be computed in $\tilde{O}(\nnz(A))$-time by the power method.}
\end{restatable}

We obtain this result by improving the main technique of~\cite{achlioptas2013near}. 
Their sampling distribution arises from optimizing a concentration bound, called the matrix-Bernstein inequality, for the sum of matrices formed by sampling entries independently. 
Our distribution is obtained by the same approach, but arises from considering the columns and rows simultaneously.

In addition to the sampling scheme in Theorem \ref{thm:spectralSparsifier}, we analyze $\ell_1$-sampling from every row (in Section \ref{sec:SecondSampling}).\footnote{Sampling entry $A_{ij}$ with probability proportional to $|A_{ij}|/\|A_i\|_1$}
This gives a worse bound than the above bound, roughly replacing the $\sr(A)$ term with $n$, but has the added advantage that the sampled matrix has uniform row-sparsity.

\paragraph{Lower Bound.}
Our next theorem complements our main result with a lower bound on the sparsity of \emph{any} $\epsilon$-spectral-norm approximation of a matrix $A$ in terms of its numerical sparsity $\ns(A)$ and error parameter $\epsilon > 0$.%
\footnote{We write $\tilde{O}(f)$ as a shorthand for $O(f \cdot \polylog(nm))$ where $n$ and $m$ are the dimensions of the matrix, and write $O_\epsilon(\cdot)$ when the hidden constant may depend on $\epsilon$.}

\begin{restatable}{theorem}{sparsifierLowerBound}\label{thm:lowerBound} 
Let $0 < \epsilon <\frac{1}{2}$ and $n,k \geq 1$ be parameters satisfying $k \leq O(\epsilon^2 n\log^2 \frac{1}{\epsilon})$.
Then, there exists a matrix $A\in\mathbb{R}^{n\times n}$ such that $\ns(A) = \Theta(k\log^2\frac{1}{\epsilon})$ and, for every matrix $B$ satisfying $\|A-B\|_2\leq \epsilon\|A\|_2$, the sparsity of every row and every column of $B$ is at least 
$\Omega(\epsilon^{-2}k\log^{-2}\frac{1}{\epsilon}) = \tilde{\Omega} (\epsilon^{-2})\cdot \ns(A)$.
\end{restatable}

While the lower bound shows that the worst-case dependence on the parameters $\ns(A)$ and $\epsilon$ is optimal, it is based on a matrix with stable rank $\Omega(n)$. Settling the sample complexity when the stable rank is $o(n)$ is an interesting open question that we leave for future work.

\subsection{Comparison to Previous Work}
The work of \cite{achlioptas2007fast} initiated a long line of work on entrywise sampling schemes that approximate a matrix under spectral-norm \citep{arora2006fast, gittens2009error, drineas2011note, kundu2014note,kundu2017recovering,nguyen2015tensor, achlioptas2013near}. Sampling entries independently has the advantage that the output matrix can be seen as a sum of independent random matrices whose spectral-norm can be bounded using known matrix concentration bounds. All previous work uses such matrix concentration bounds with the exception of \cite{arora2006fast} who bound the spectral-norm of the resulting matrix by analyzing the Rayleigh quotient of all possible vectors. 

Natural distributions to sample entries are the $\ell_2$ and $\ell_1$ distributions, which correspond to sampling entry $A_{ij}$ with probability proportional to $A_{ij}^2/\|A\|_F^2$ and $|A_{ij}|/\|A\|_1$ respectively.%
\footnote{Here and henceforth we denote by $\|A\|_1$ the entry-wise $l_1$ norm.}

Prior work that use variants of the $\ell_2$ sampling \citep{achlioptas2007fast,drineas2011note,nguyen2015tensor, kundu2014note} point out that sampling according to the $\ell_2$ distribution causes small entries to ``blow-up" when sampled. Some works, e.g.~\cite{drineas2011note}, get around this by zeroing-out small entries or by exceptional handling of small entries, e.g.~\cite{achlioptas2007fast}, while others used distributions that combine the $\ell_1$ and $\ell_2$ distributions, e.g.~\cite{kundu2014note}. All these works sample $\Omega(\epsilon^{-2}n\sr(A))$ entries in expectation to achieve an $\epsilon$-spectral-norm approximation and our Theorem \ref{thm:spectralSparsifier} provides an asymptotically better bound. 
For a full comparison see Table \ref{table:sparsificationComparison}.

All these algorithms, including the algorithm of Theorem \ref{thm:spectralSparsifier}, sample a number of entries corresponding to $\sr(A)$, hence they must have an estimate of it, which requires estimating $\|A\|_2$. 
An exception is the bound in Theorem \ref{thm:sparsifierForAMM}, which can be achieved without this estimate.
In practice, however, and in previous work in this area, there is a sampling budget $s \geq 0$ and $s$ samples are drawn according to the stated distribution, avoiding the need for this estimate. In this case, the algorithm of Theorem \ref{thm:spectralSparsifier} can be implemented in two-passes over the data and in $O(\nnz(A))$ time.

\ifcolt
\renewcommand{\arraystretch}{1.25}
\begin{table*}[!t]
\caption{\label{table:sparsificationComparison} Comparison between schemes for $\epsilon$-spectral-norm sparsification. 
The first two entries in the third column present the ratio between the referenced sparsity and that of Theorem~\ref{thm:spectralSparsifier}. 
} 
\begin{center}
\begin{tabulary}{\textwidth}{|p{0.35\textwidth} > {\raggedright}p{0.27\textwidth} >  {\raggedright\arraybackslash}p{0.30\textwidth}|}
\hline
\textbf{Expected Number of Samples} &{\textbf{Reference}}     & \textbf{Compared to Thm.~\ref{thm:spectralSparsifier}}\\ 
\hline
\hline
$O(\epsilon^{-1}n\sqrt{\ns(A)\sr(A)})$            &\cite{arora2006fast}  & $\tilde{O}_\epsilon \Big(\min\Big( \tfrac{n}{\sqrt{\ns(A)\sr(A)}} , \sqrt{{n}} \Big)\Big)$ \\ 
 $O(\epsilon^{-2}n\sr(A) + n\polylog(n))$  &\cite{achlioptas2007fast} & $\tilde{O}_\epsilon \Big(\min\Big( \tfrac{n}{\ns(A)} , \sqrt{\tfrac{n\sr(A)}{\ns(A)}} \Big)\Big)$\\
$\tilde{O}(\epsilon^{-2}n\sr(A))$ &\cite{drineas2011note,kundu2014note} &  \\
\hline
$\tilde{O}({\epsilon^{-2}}\ns(A) \sr(A) +$ $\epsilon^{-1}\sqrt{\ns(A) \sr(A)n})$ & \cite{achlioptas2013near}; Theorem \ref{thm:spectralSparsifier}  
& \cite{achlioptas2013near} is only for data matrices\\ 
\hline
$\tilde{O}(\epsilon^{-2}n\ns(A))$ & Theorem \ref{thm:sparsifierForAMM} & bounded row-sparsity \\
\hline\hline
$\Omega(\epsilon^{-2}n\ns(A) \log^{-4}\frac{1}{\epsilon})$ &Theorem \ref{thm:lowerBound} &$\sr(A)=\Theta(n)$ \\
\hline
\end{tabulary}
\end{center}
\end{table*}

\else
\renewcommand{\arraystretch}{1.25}
\begin{table*}[!t]
\caption{\label{table:sparsificationComparison} Comparison between schemes for $\epsilon$-spectral-norm sparsification. 
The first two entries in the third column present the ratio between the referenced sparsity and that of Theorem~\ref{thm:spectralSparsifier}. 
} 
\begin{center}
\begin{tabulary}{\textwidth}{|p{0.39\textwidth} > {\raggedright}p{0.27\textwidth} >  {\raggedright\arraybackslash}p{0.27\textwidth}|}
\hline
\textbf{Expected Number of Samples} &{\textbf{Reference}}     & \textbf{Compared to Thm.~\ref{thm:spectralSparsifier}}\\ 
\hline
\hline
$O(\epsilon^{-1}n\sqrt{\ns(A)\sr(A)})$            &\cite{arora2006fast}  & $\tilde{O}_\epsilon \Big(\min\Big( \tfrac{n}{\sqrt{\ns(A)\sr(A)}} , \sqrt{{n}} \Big)\Big)$ \\ 
 $O(\epsilon^{-2}n\sr(A) + n\polylog(n))$  &\cite{achlioptas2007fast} & $\tilde{O}_\epsilon \Big(\min\Big( \tfrac{n}{\ns(A)} , \sqrt{\tfrac{n\sr(A)}{\ns(A)}} \Big)\Big)$\\
$\tilde{O}(\epsilon^{-2}n\sr(A))$ &\cite{drineas2011note,kundu2014note} &  \\
\hline
$\tilde{O}({\epsilon^{-2}}\ns(A) \sr(A) + \epsilon^{-1}\sqrt{\ns(A) \sr(A)n})$ & \cite{achlioptas2013near}; Theorem \ref{thm:spectralSparsifier}  
& \cite{achlioptas2013near} is only for data matrices\\ 
\hline
$\tilde{O}(\epsilon^{-2}n\ns(A))$ & Theorem \ref{thm:sparsifierForAMM} & bounded row-sparsity \\
\hline\hline
$\Omega(\epsilon^{-2}n\ns(A) \log^{-4}\frac{1}{\epsilon})$ &Theorem \ref{thm:lowerBound} &$\sr(A)=\Theta(n)$ \\
\hline
\end{tabulary}
\end{center}
\ifaistats \else  \fi
\end{table*}
\fi

\subsection{Applications of Spectral-Norm Sparsification}
We provide two useful applications of spectral-norm sparsification.
More precisely, we use the sparsification to speed up two computational tasks on numerically sparse matrices: 
approximate matrix multiplication and approximate ridge regression. 
This adds to previous work, which showed applications 
to low-rank approximation \citep{achlioptas2007fast}, 
to semidefinite programming \citep{arora2006fast}, 
and to PCA and sparse PCA \citep{kundu2017recovering}. 
These applications work in a black-box manner, 
and can thus employ our improved sparsification scheme.

\paragraph{Application \RNum{1}: Approximate Matrix Multiplication (AMM).} Given matrices $A\in\mathbb{R}^{m\times n},B\in\mathbb{R}^{n\times p}$ and error parameter $\epsilon>0$, the goal is to compute a matrix $C \in\mathbb{R}^{m\times p}$ such that $\|AB-C\|\leq \epsilon \|A\|\cdot\|B\|$, where the norm is usually either Frobenius-norm $\|\cdot\|_F$ or spectral-norm $\|\cdot\|_2$. In Section \ref{sec:AMM}, we provide algorithms for both error regimes by combining our entrywise sampling scheme with previous AMM algorithms that sample a small number of columns of $A$ and rows of $B$.  

\begin{restatable}{theorem}{AMMspectralNorm}\label{thm:AMM_spectralNorm} 
There exists an algorithm that,
given matrices $A\in\mathbb{R}^{m\times n},B\in\mathbb{R}^{n\times p}$ parameter $0<\epsilon <\frac{1}{2}$ and constant factor estimates of $\|A\|_2$ and $\|B\|_2$, computes a matrix $C\in\mathbb{R}^{m\times p}$ satisfying with high probability $\|AB-C\|_2\leq \epsilon\|A\|_2\|B\|_2$ in time \ifaistats $O(\nnz(A)+\nnz(B)) + \tilde{O}(\epsilon^{-6}\sqrt{\sr(A)\sr(B)}\ns(A)\ns(B)).$
\else
$$O(\nnz(A)+\nnz(B)) + \tilde{O}(\epsilon^{-6}\sqrt{\sr(A)\sr(B)}\ns(A)\ns(B)).$$
\fi
\end{restatable}

\begin{restatable}{theorem}{AMMfrobeniusNorm}\label{thm:AMM_frobeniusNorm}
There exists an algorithm that,
given matrices $A\in\mathbb{R}^{m\times n},B\in\mathbb{R}^{n\times p}$ and parameter $0<\epsilon<\frac{1}{2}$, 
computes a matrix $C\in\mathbb{R}^{m\times p}$ satisfying $\E\|AB-C\|_F\leq \epsilon\|A\|_F\|B\|_F$ in time \ifaistats $O(\nnz(A)+\nnz(B)+ \epsilon^{-6}\ns(A)\ns(B)).$
\else
$$O(\nnz(A)+\nnz(B)+ \epsilon^{-6}\ns(A)\ns(B)).$$
\fi
\end{restatable}

Approximate Matrix Multiplication (AMM) is a fundamental problem in numerical linear algebra with a long line of formative work \citep{frieze2004fast,drineas2006fast,clarkson2009numerical,magen2011low,cohen2015optimal,ye2016frequent,mroueh2017co} 
and many others. These results fall into roughly three categories; sampling based methods, random projection based methods and a mixture of sampling and projection based methods. We focus on sampling based methods in our work. 

There are two main error regimes considered in the literature: spectral-norm error and Frobenius-norm error.  We focus on the results of \cite{magen2011low} for spectral-norm error and \cite{drineas2006fast} for Frobenius-norm error.
Sampling based methods, including that of \cite{drineas2006fast,magen2011low}, propose sampling schemes that are linear time or nearly-linear time: specifically, they write the product of two matrices as the sum of $n$ outer products $AB=\sum_{i \in [n]}A^iB_i$, and then sample and compute each outer product $A^iB_i/p_i$ with probability $p_i \propto \|A^i\|_2\|B_i\|_2$. Computing each of these rank-1 outer products takes time bounded by $O(\csp(A)\rsp(B))$. This estimator is repeated sufficiently many times depending on the error regime under consideration. 

Our entrywise-sampling scheme compounds well with this framework for approximate matrix multiplication by additionally sampling entries from the rows/columns sampled by the AMM algorithm. We essentially replace the $\csp(A)\rsp(B)$ term with $\ns(A)\ns(B)$, up to $\tilde{O}(\poly(1/\epsilon))$ factors, for both Frobenius-norm and spectral-norm error regimes. 
It is plausible that the dependence on epsilon can be improved.

\paragraph{Application \RNum{2}: Approximate Ridge Regression.}
Given a matrix $A\in\R^{m\times n}$, a vector $b \in \R^m$ and a parameter $\lambda > 0$, the goal is to find a vector $x \in \R^n$ that minimizes $\|Ax-b\|_2^2+\lambda\|x\|_2^2$. This problem is $\lambda$-strongly convex, has solution $x^*=(A^\top A + \lambda I)^{-1}A^\top b$ and condition number $\kappa_\lambda(A^\top A) \eqdef {\|A\|_2^2}/{\lambda}$. 

Given an initial vector $x_0\in\mathbb{R}^n$ and a parameter $\epsilon > 0$, an $\epsilon$-approximate solution to the ridge regression problem is a vector $\hat{x} \in \R^n$ satisfying $\|\hat{x}-x^*\|_{A^\top A+\lambda I}\leq \epsilon\|x_0-x^*\|_{A^\top A+\lambda I}$,
where we write $\|x\|_M\eqdef x^\top M x$ when $M$ is a PSD matrix. 
We provide algorithms 
for approximate ridge regression by using our sparsification scheme as a preconditioner for known linear-system solvers in composition with a black-box acceleration framework by \cite{frostig2015regularizing}.
The following theorem is proved in 
\iffull Section \ref{sec:PreconRegProblem}. \else the full version. \fi

\begin{restatable}{theorem}{generalRowNormsRR}\label{thm:mainRidgeRegression}
There exists an algorithm that,
given $A\in \R^{m \times n},x_0\in\R^n,\lambda>0$ and $\epsilon>0$, 
computes with high probability an $\epsilon$-approximate solution to the ridge regression problem
in time 
\ifaistats $
O_\epsilon(\nnz(A)) + \tilde{O}_\epsilon\left((\nnz(A))^{2/3}(\ns(A) \sr(A))^{1/3}\sqrt{\kappa_\lambda(A^\top A)}\right).
$ 
\else $$
O_\epsilon(\nnz(A)) + \tilde{O}_\epsilon\left((\nnz(A))^{2/3}(\ns(A) \sr(A))^{1/3}\sqrt{\kappa_\lambda(A^\top A)}\right).
$$ \fi
\end{restatable}

Moreover, when the input matrix $A$ has uniform column (or row) norms, the running time in Theorem \ref{thm:mainRidgeRegression} can be reduced by a factor of roughly $(\sr(A)/n)^{1/6}$, 
for details see
\iffull Section \ref{sec:RRUnifromNorms}. \else the full version. \fi

Solving linear systems using preconditioning has a rich history that is beyond the scope of this work to summarize. Recently, the work of \cite{gupta2018exploiting} designed algorithms with improved running times over popular methods using the Stochastic Variance Reduced Gradient Descent (SVRG) framework of \cite{johnson2013accelerating}. They adapt it using efficient subroutines for numerically sparse matrices. They also suggested the idea of using spectral-norm sparsifiers as preconditioners for linear regression.
While they considered the sparsification of \cite{achlioptas2013near} for computing the preconditioners, they required a stronger bound on the spectral-norm approximation than Theorem \ref{thm:mainRidgeRegression} does. 

Our result is in general incomparable to that of \cite{gupta2018exploiting}. In the case when the input has uniform column (or row) norms, our running time is roughly an $(\ns(A)/n)^{1/6}$-factor smaller than theirs, for details see 
\iffull Theorem~\ref{thm:uniformRowNormsRR} in Section~\ref{sec:RRUnifromNorms}. 
\else the full version. 
\fi

Very recently, \cite{carmon2020coordinate} have developed, independently of our work and as part of a suite of results on bilinear minimax problems, an algorithm for ridge regression with improved running time $\tilde{O}(\nnz(A) + \sqrt{\nnz(A) \ns(A) \sr(A) \kappa_\lambda(A^\top A)})$. 
Their approach is different and their techniques are more involved than ours.

\section{\ifaistats SPECTRAL-NORM SPARSIFICATION \else Spectral-Norm Sparsification \fi}
In this section we state and prove our main results. We first prove the upper bound in Theorem \ref{thm:spectralSparsifier}. 
Then we analyze $\ell_1$ sampling from the rows in Theorem \ref{thm:sparsifierForAMM}, Section \ref{sec:SecondSampling} that gives a slightly weaker bound but has the property that the resulting matrix has uniform row sparsity. 
In Section \ref{sec:lowerBound}, we prove the lower bound in Theorem \ref{thm:lowerBound}.

\sparsifier*

Before we prove Theorem \ref{thm:spectralSparsifier}, we start by stating a result on the concentration of sums of independent random matrices; the Matrix Bernstein Inequality.

\begin{theorem}[Matrix Bernstein, Theorem 1.6 of \cite{tropp2012user}]\label{thm:matrixBernstein}
Consider a finite sequence $\{Z_k\}$ of independent, random $d_1\times d_2$ real matrices, such that there is $R>0$ satisfying $\E Z_k =0$ and $\|Z_k\|_2 \leq R$ almost surely. Define
\[
  \sigma^2 
  = \max \Big\{\Big\| \sum_k \E(Z_k Z_k^\top)\Big\|_2, \Big\|\sum_k \E(Z_k^\top Z_k ) \Big\|_2 \Big\}.
  \]
Then for all $t\geq 0$,
\[
\mathbb{P} \Big(\Big\|\sum_k Z_k \Big\|_2\geq t \Big)\leq
(d_1 + d_2 ) \exp\Big(\frac{-t^2/2}{\sigma^2+Rt/3}\Big).
\]
\end{theorem}

\ifcolt
\begin{proof}\textbf{of Theorem~\ref{thm:spectralSparsifier}. }
\else%
\begin{proof}[Proof of Theorem~\ref{thm:spectralSparsifier}. ]
\fi
Let $\epsilon>0$. 
Given a matrix $A$, define sampling probabilities as follows. 
\begin{empheq}[box=\widefbox]{align*}
p_{ij}^{(1)}&=\frac{|A_{ij}|}{\sum_{i'j'}|A_{i'j'}|}\\
p_{ij}^{(2)}&=\frac{\|A_i\|_1^2}{\sum_{i'}\|A_{i'}\|_1^2}\cdot\frac{|A_{ij}|}{\|A_i\|_1}\\
p_{ij}^{(3)}&=\frac{\|A^j\|_1^2}{\sum_{j'}\|A^{j'}\|_1^2}\cdot\frac{|A_{ij}|}{\|A^j\|_1}\\
p^*_{ij}&=\max_\alpha (p_{ij}^{(\alpha)}) .
\end{empheq}
Observe that each $\alpha=1,2,3$ yields a probability distribution
because $\sum_{ij} p_{ij}^{(\alpha)}=1$. 

Let $s<mn$ be a parameter that we will choose later.
Now sample each entry of $A$ independently 
and scale it to get an unbiased estimator,
i.e., compute $\tilde{A}$ by 
\begin{empheq}[box=\widefbox]{align*}
  \tilde{A}_{ij}=\begin{cases}
  \frac{A_{ij}}{p_{ij}} 
  & \text{with prob.\ $p_{ij}=\min(1,s\cdot p^*_{ij})$; } \\
  0 
  & \text{otherwise.}
  \end{cases}
\end{empheq}
To bound the expected sparsity, 
observe that $p^*_{ij}\leq\sum_\alpha p_{ij}^{(\alpha)}$,
and thus
\[
  \E[\nnz(\tilde{A})] 
  =    \sum_{ij} p_{ij} 
  \leq s \sum_{ij} \sum_\alpha p_{ij}^{(\alpha)}
  \leq 3s.
\]
We show that each of the above distributions bounds one of the terms in matrix Bernstein bound.
For each pair of indices $(i,j)$ define a matrix $Z_{ij}$ that has a single non-zero at the $(i,j)$ entry, with value $\tilde{A}_{ij}-A_{ij}$. 
Its spectral-norm is $\|Z_{ij}\|_2=|\tilde{A}_{ij}-A_{ij}|$. If $p_{ij}=1$, this is 0. If $p_{ij}<1$ then
\ifaistats
\begin{align*}
  &|\tilde{A}_{ij}-A_{ij}|
  \leq |A_{ij}|\max(1,\frac{1}{p_{ij}}-1)\\
  &\leq \frac{|A_{ij}|}{p_{ij}}
  \leq \frac{|A_{ij}|}{sp_{ij}^{(1)}}
  = \frac{1}{s}\sum_{i'j'}|A_{i'j'}|
  \\
  &\leq \frac{\sqrt{\ns(A)}}{s}\sum_{j}\|A^{j}\|_2
  \leq \frac{\sqrt{\ns(A)n}}{s}\|A\|_F
  \eqqcolon R,
\end{align*}
\else
\begin{align*}
  |\tilde{A}_{ij}-A_{ij}|
  &\leq |A_{ij}|\max(1,\frac{1}{p_{ij}}-1)\\
  &\leq \frac{|A_{ij}|}{p_{ij}}
  \leq \frac{|A_{ij}|}{sp_{ij}^{(1)}}
  = \frac{1}{s}\sum_{i'j'}|A_{i'j'}|
  \\
  &\leq \frac{\sqrt{\ns(A)}}{s}\sum_{j}\|A^{j}\|_2
  \leq \frac{\sqrt{\ns(A)n}}{s}\|A\|_F
  \eqqcolon R,
\end{align*}
\fi
where the last inequality follows from Cauchy-Schwarz inequality.

In order to bound $\sigma^2$, first notice that $\var(\tilde{A}_{ij})\leq\E(\tilde{A}_{ij}^2)=\frac{A_{ij}^2}{sp^*_{ij}}$.
Now, since $Z_{ij}Z_{ij}^\top$ has a single non-zero entry at $(i,i)$, and $Z_{ij}^\top Z_{ij}$ has a single non-zero entry at $(j,j)$, both $\sum_{i,j}Z_{ij}Z_{ij}^\top$ and $\sum_{i,j}Z_{ij}^\top Z_{ij}$ are diagonal, where the $(i,i)$ entry is $\sum_j (\tilde{A}_{ij}-A_{ij})^2$ in the former and the $(j,j)$ entry is $\sum_i (\tilde{A}_{ij}-A_{ij})^2$ in the latter. Since these are diagonal matrices, their spectral-norm equals their largest absolute entry, and thus 
\ifaistats
\begin{align*}
&\Big\|\sum_{i,j}\E\big(Z_{ij}Z_{ij}^\top\big)\Big\|_2 
    \leq\max_i\Big(\sum_j\frac{A_{ij}^2}{sp^*_{ij}}\Big)
    \\    
    &\leq\max_i\Big(\sum_j\frac{A_{ij}^2}{sp^{(2)}_{ij}}\Big)=\frac{1}{s}\max_i\Big(\sum_j\frac{|A_{ij}|\sum_{i'}\|A_{i'}\|_1^2}{\|A_i\|_1}\Big)
    \\
    &=\frac{1}{s}\sum_{i'}\|A_{i'}\|_1^2\leq\frac{1}{s}\sum_{i'}\ns(A)\|A_{i'}\|_2^2=\frac{\ns(A)}{s}\|A\|_F^2.
\end{align*}
\else
\begin{align*}
\Big\|\sum_{i,j}\E\big(Z_{ij}Z_{ij}^\top\big)\Big\|_2 
    & \leq\max_i\Big(\sum_j\frac{A_{ij}^2}{sp^*_{ij}}\Big)
    \leq\max_i\Big(\sum_j\frac{A_{ij}^2}{sp^{(2)}_{ij}}\Big)
    \\
    &=\frac{1}{s}\max_i\Big(\sum_j\frac{|A_{ij}|\sum_{i'}\|A_{i'}\|_1^2}{\|A_i\|_1}\Big)
    =\frac{1}{s}\sum_{i'}\|A_{i'}\|_1^2
    \\
    &\leq\frac{1}{s}\sum_{i'}\ns(A)\|A_{i'}\|_2^2=\frac{\ns(A)}{s}\|A\|_F^2.
\end{align*}
\fi

The same bound can be shown for $\sum_{i,j}\E (Z_{ij}^\top Z_{ij})$ by using ${p^*_{ij}} \geq {p^{(3)}_{ij}}$, thus by the definition of $\sigma^2$, $\sigma^2\leq\frac{\ns(A)}{s}\|A\|_F^2$. Finally, by the matrix-Bernstein bound,
\ifaistats \begin{align*}
&\mathbb{P}\big(\big\|\sum_{i,j}Z_{ij}\big\|_2 \geq\epsilon\|A\|_2\big) \\ 
&\leq 2m \exp\bigg(-\frac{\epsilon^2\|A\|_2^2/2}{\frac{\ns(A)}{s}\|A\|_F^2+\epsilon\frac{\sqrt{\ns(A)n}}{s}\|A\|_F\|A\|_2/3}\bigg),
\end{align*}
\else 
\begin{align*}
\mathbb{P}\Big(\Big\|\sum_{i,j}Z_{ij}\Big\|_2 \geq\epsilon\|A\|_2\Big) 
\leq 2m \exp\bigg(-\frac{\epsilon^2\|A\|_2^2/2}{\frac{\ns(A)}{s}\|A\|_F^2+\epsilon\frac{\sqrt{\ns(A)n}}{s}\|A\|_F\|A\|_2/3}\bigg),
\end{align*}
\fi
and since $\sr(A)=\frac{\|A\|_F^2}{\|A\|_2^2}$, by setting $s=O(\epsilon^{-2}\ns(A) \sr(A)\log m+\epsilon^{-1}\sqrt{\ns(A)\cdot n \cdot \sr(A)}\log m)$ 
we conclude that with high probability $\|\tilde{A}-A\|_2\leq\epsilon\|A\|_2$, which completes the proof of Theorem \ref{thm:spectralSparsifier}.
\end{proof}

\subsection{A Second Sampling Scheme}
\label{sec:SecondSampling} 
We analyze $\ell_1$ row sampling, i.e. sampling entry $(i,j)$ with probability $\frac{|A_{ij}|}{\|A_i\|_1}$, as was similarly done for numerically sparse matrices in \cite{gupta2018exploiting}, although they employed this sampling (i) in a different setting and (ii) on one row at a time. Here, we analyze how to employ this sampling on all the rows simultaneously for $\epsilon$-spectral-norm sparsification.
This sampling is inferior to the one in Theorem~\ref{thm:spectralSparsifier} in terms of $\nnz(\tilde{A})$, but has the additional property that the sparsity of every row is bounded. 
By applying this scheme to $A^\top$, we can alternatively obtain an $\epsilon$-spectral-norm sparsifier where the sparsity of every column is bounded.

\begin{theorem}\label{thm:sparsifierForAMM}
There is an algorithm that, given a matrix $A\in\mathbb{R}^{m\times n}$ and a parameter $\epsilon>0$, 
computes in time $O(\nnz(A))$ with high probability an $\epsilon$-spectral-norm sparsifier $\tilde{A}$ for $A$ such that the sparsity of every row of $\tilde{A}$ is bounded by $O(\epsilon^{-2}\ns(A)\log (m+n))$. 
\end{theorem}
The algorithm is as follows.
Given a matrix $A$ and $\epsilon>0$, 
define the sampling probabilities 
\begin{empheq}[box=\widefbox]{align*}
p_{ij}=\frac{|A_{ij}|}{\|A_i\|_1} , 
\end{empheq}
and observe that for every $i$ this induces probability distribution, i.e., $\sum_j p_{ij}=1$. 
Let $s=O(\epsilon^{-2}\ns(A)\log (m+n))$. 
Now from each row of $A$ sample $s$ entries independently with replacement according to the above distribution, 
and scale it to get an unbiased estimator of that row;
formally, for each row $i$ and each $t=1,\ldots,s$ draw a row vector
\begin{empheq}[box=\widefbox]{align*}
  Q_i^{(t)} = \begin{cases}
  \frac{A_{ij}}{p_{ij}}e_j^\top 
  & \text{ with prob.\ $p_{ij}$,} 
  \end{cases}
\end{empheq}
where $\{e_j\}_j$ is the standard basis of $\mathbb{R}^n$. 
Next, average the $t$ samples for each row, and arrange these rows in a matrix $\tilde{A}$ that is an unbiased estimator for $A$; formally, 
\begin{empheq}[box=\widefbox]{align*}
  \tilde{A} = \sum_{i=1}^m e_i \frac{1}{s}\sum_{t=1}^s Q_i^{(t)}. 
\end{empheq}
Clearly $\E(\tilde{A})=A$ and every row of $\tilde{A}$ has at most $s$ non-zeros. 
In order to bound the probability that $\tilde{A}$ is an $\epsilon$-spectral-norm sparsifier of $A$, similarly to the proof of Theorem \ref{thm:spectralSparsifier}, we employ the matrix-Bernstein bound stated in Theorem~\ref{thm:matrixBernstein}. 
\iffull


\ifcolt
\begin{proof}\textbf{of Theorem \ref{thm:sparsifierForAMM}.}
\else
\begin{proof}[Proof of Theorem \ref{thm:sparsifierForAMM}.] 
\fi
\label{proof:secondSampling}
Given a matrix $A,\epsilon>0$, let $k=\ns(A)$ and apply the algorithm of Theorem \ref{thm:sparsifierForAMM}.
Note that by the definition of $\ns(A)$ and by spectral-norm properties, the $i$-th row of $A$ satisfies
\begin{equation}\label{eq:rowNormToSpectral}
  \|A_i\|_1\leq\sqrt{k}\|A_i\|_2\leq \sqrt{k}\|A\|_2.  
\end{equation}
For each random draw, define a matrix $Z_{(it)}$ with exactly one non-zero row 
formed by placing $A_i - Q_i^{(t)}$ at the $i$-th row; 
formally, let $Z_{(it)}=e_i (A_i - Q_i^{(t)})$. Where it is clear from context we will omit the superscript from $Q_i^{(t)}$.
The spectral-norm of $Z_{(it)}$ is 
\ifaistats
\begin{align*}
  \|Z_{(it)}\|_2
  = \|A_i - Q_i^{(t)}\|_2
  \leq \|A_i\|_2 + \|Q_i^{(t)}\|_2 \\
  = \|A_i\|_2 + \|A_i\|_1
  \leq 2\sqrt{k}\|A\|_2
  \eqqcolon R.
\end{align*}
\else
\[ \|Z_{(it)}\|_2
  = \|A_i - Q_i^{(t)}\|_2
  \leq \|A_i\|_2 + \|Q_i^{(t)}\|_2 \\
  = \|A_i\|_2 + \|A_i\|_1
  \leq 2\sqrt{k}\|A\|_2
  \eqqcolon R.\]
\fi

To bound $\sigma^2$, notice that $Z_{(it)} Z_{(it)}^\top$ has a single non-zero at the $(i,i)$ entry with value $\|A_i - Q_i^{(t)}\|_2^2$, hence 
\ifaistats 
    $\big\|\E \sum_{i,t} Z_{(it)} Z_{(it)}^\top\big\|_2 
      = s \max_i \E\|A_i - Q_i\|_2^2= s \max_i \E\|Q_i\|_2^2-\|A_i\|_2^2 \leq s \max_i \sum_j \|A_i\|_1\cdot|A_{ij}|
      \leq s k\|A\|_2^2$.
\else 
    \begin{align*}
      \big\|\E \sum_{i,t} Z_{(it)} Z_{(it)}^\top\big\|_2 
      &= s \max_i \E\|A_i - Q_i\|_2^2= s \max_i \E\|Q_i\|_2^2-\|A_i\|_2^2 \\ 
      &\leq s \max_i \sum_j \|A_i\|_1\cdot|A_{ij}|
      \leq s k\|A\|_2^2.
    \end{align*}
\fi

The other term $Z_{(it)}^\top Z_{(it)}$ satisfies $\E(Z_{(it)}^\top Z_{(it)}) = \E\big(Q_i^{\top} (Q_i - A_i)\big) = \E(Q_i^{\top} Q_i) - A_i^\top A_i$. The matrix $\E(Q_i^{\top} Q_i)$ is diagonal with value $|A_{ij}|\cdot\|A_i\|_1$ at the $(j,j)$ entry, hence
\ifaistats
\begin{align*}
    &\big\|\sum_{i,t} \E(Z_{(it)}^\top Z_{(it)})\big\|_2 = s\big\|\sum_{i} (\E(Q_i^{\top} Q_i) - A_i^\top A_i)\big\|_2\\
    & = s\|\sum_{i} \E(Q_i^{\top} Q_i) - A^\top A\|_2 \\
    & \leq s\big(\big\|\sum_{i} \E(Q_i^{\top} Q_i)\big\|_2 + \|A^\top A\|_2\big) \\
    & = s\big(\max_j \sum_{i} |A_{ij}|\cdot\|A_i\|_1 + \|A\|_2^2\big) \\
    & \leq s\sqrt{k}\big(\|A\|_2\max_j \sum_{i} |A_{ij}| + \|A\|_2^2\big)\\
    & = s\sqrt{k}\big(\|A\|_2\max_j \|A^{j}\|_1 + \|A\|_2^2\big)\leq 2s\cdot k\cdot \|A\|_2^2\eqqcolon \sigma^2.
\end{align*}
\else
\begin{align*}
    \big\|\sum_{i,t} \E(Z_{(it)}^\top Z_{(it)})\big\|_2 &= s\big\|\sum_{i} (\E(Q_i^{\top} Q_i) - A_i^\top A_i)\big\|_2\\
    & = s\|\sum_{i} \E(Q_i^{\top} Q_i) - A^\top A\|_2 \\
    & \leq s\big(\big\|\sum_{i} \E(Q_i^{\top} Q_i)\big\|_2 + \|A^\top A\|_2\big) \\
    & = s\big(\max_j \sum_{i} |A_{ij}|\cdot\|A_i\|_1 + \|A\|_2^2\big) \\
    & \leq s\sqrt{k}\big(\|A\|_2\max_j \sum_{i} |A_{ij}| + \|A\|_2^2\big)\\
    & = s\sqrt{k}\big(\|A\|_2\max_j \|A^{j}\|_1 + \|A\|_2^2\big)\leq 2s\cdot k\cdot \|A\|_2^2\eqqcolon \sigma^2.
\end{align*}
\fi
Now, by the matrix-Bernstein bound as stated in Theorem \ref{thm:matrixBernstein},
\ifaistats
\begin{align*}
  &\mathbb{P}(\|A-\tilde{A}\|_2\geq\epsilon\|A\|_2) = \mathbb{P}\big(\big\|\sum_{i,t}Z_{(it)}\big\|_2\geq s\epsilon\|A\|_2\big)  \\
  & \leq(m+n)\exp\Big(-\frac{s\epsilon^2\|A\|_2^2/2}{2k\|A\|_2^2+\tfrac{2\epsilon}{3}\sqrt{k}\|A\|_2^2}\Big),
\end{align*}
\else
\begin{align*}
  \mathbb{P}(\|A-\tilde{A}\|_2\geq\epsilon\|A\|_2) &= \mathbb{P}\big(\big\|\sum_{i,t}Z_{(it)}\big\|_2\geq s\epsilon\|A\|_2\big)  \\
  & \leq(m+n)\exp\Big(-\frac{s\epsilon^2\|A\|_2^2/2}{2k\|A\|_2^2+\tfrac{2\epsilon}{3}\sqrt{k}\|A\|_2^2}\Big),
\end{align*}
\fi

and by setting $s=O(\epsilon^{-2}k\log (m+n))$ we conclude that with high probability $\|\tilde{A}-A\|_2\leq\epsilon\|A\|_2$.
\end{proof}

\else
The proof is omitted here and appears in the full version.
\fi

\subsection{Lower Bounds}\label{sec:lowerBound}

We provide a lower bound in Theorem \ref{thm:lowerBound} for spectral-norm sparsification, which almost matches the bound in Theorem~\ref{thm:spectralSparsifier}
for a large range of $\epsilon$ and $\ns(A)$.  

\sparsifierLowerBound*

\begin{proof}
We shall assume that $k$ divides $n$, 
and that both are powers of $2$, 
which can be obtained with changing the bounds by a constant factor. 
Let $m=\frac{n}{k}$, and notice it is a power of 2 as well. 

Construct first a vector $a\in\mathbb{R}^m$ 
by concatenating blocks of length $2^i$ whose coordinates have value $2^{-(1+\alpha)i}$, for each $i\in\{0,...,\log m-1\}$, 
where $1>\alpha \geq \Omega(\log^{-1} m)$ is a parameter that we will set later.
The last remaining coordinate have value $0$.
Formally, the coordinates of $a$ are given by 
$a_j=2^{-(1+\alpha) \lfloor\log j\rfloor}$, except the last one which is $0$. 
\ifaistats
Its $\ell_1$ norm is $\|a\|_1
  = \sum_{j=1}^m a_j=\sum_{i=0}^{\log m-1} 2^i\cdot 2^{-(1+\alpha)i}
  = \frac{1-2^{-\alpha\log m}}{1-2^{-\alpha}}
  = \Theta(\alpha^{-1}).$
\else
Its $\ell_1$ norm is
\[
  \|a\|_1
  = \sum_{j=1}^m a_j=\sum_{i=0}^{\log m-1} 2^i\cdot 2^{-(1+\alpha)i}
  = \frac{1-2^{-\alpha\log m}}{1-2^{-\alpha}}
  = \Theta(\alpha^{-1}) .
\]
\fi
A similar computation shows that $\|a\|_2=\Theta(1)$, 
and thus $\ns(a) = \Theta(\alpha^{-2})$. 
Denote by $a_{\tail(c)}$ the vector $a$ without its $c$ largest entries,
then its $\ell_2$ norm is 
\begin{equation} \label{eq:c_tail}
  \|a_{\tail(c)}\|_2^2
  \geq \sum_{\mathclap{i=\lfloor\log c\rfloor+1}}^{\log m-1} 2^i\cdot 2^{-2(1+\alpha)i}
  = \Omega(c^{-(1+2\alpha)}), 
\end{equation}
which almost matches the upper bound of 
Lemma 3 in~\cite{gupta2018exploiting}.

Now, for $k=1$ we construct a circulant matrix $A\in\mathbb{R}^{m\times m}$ 
by letting the vector $a$ be its first row, and the $j$-th row is a cyclic shift of $a$ with offset $j$.
By well-known properties of circulant matrices, the $t$-th eigenvalue of $A$ is given by $\lambda_t=\sum_j a_j(\omega_t)^j$ where $\omega_t=\exp\brackets{i\frac{2\pi t}{m}}$ and $i$ is the imaginary unit, so $\|A\|_2=\|a\|_1=\Theta(\alpha^{-1})$. 
Consider $B\in\mathbb{R}^{m\times m}$ 
satisfying $\|A-B\|_2\leq \epsilon\|A\|_2$,
and suppose some row $B_j$ of $B$ has $s$ non-zeros. 
Then using~\eqref{eq:c_tail}, 
\[
  \|A-B\|_2
  \geq \|A_j-B_j\|_2
  \geq \|a_{\tail(s)}\|_2
  = \Omega(s^{-(\frac{1}{2}+\alpha)}). 
\]
By the error bound $\|A-B\|_2\leq \epsilon\|A\|_2$,
we must have 
$
  s 
  \geq (\Omega(\epsilon/\alpha))^{-\frac{2}{1+2\alpha}}
  \geq \Omega((\epsilon/\alpha)^{-\frac{2}{1+2\alpha}})$, 
which bounds from below the sparsity of every row,
and similarly also of every column, of $B$. 

To generalize this to larger numerical sparsity, 
consider as a first attempt constructing a vector $a'\in\mathbb{R}^n$ by concatenating $k$ copies of $a$. Then clearly $\ns(a')=\Theta(k \ns(a))$.
The circulant matrix of $a'$ is equivalent to $A\otimes C$, where $C$ is the all-ones matrix of dimension $k\times k$, and $\otimes$ is the Kronecker product. But this matrix has low rank, and thus might be easier to approximate. 
We thus construct a different matrix $A'=A\otimes H_k$, where $H_k$ is the $k\times k$ Hadamard matrix. Its numerical sparsity is the same as of the vector $a'$, thus $\ns(A')=\Theta(k \ns(a))$. 
The eigenvalues of $H_k$ are $\pm\sqrt{k}$. By properties of the Kronecker product, every eigenvalue of $A'$ is the product of an eigenvalue of $A$ with $\pm\sqrt{k}$, thus $\|A'\|_2=\Theta(\sqrt{k}\|A\|_2)=\Theta(\sqrt{k}\alpha^{-1})$. 
We now apply the same argument we made for $k=1$. 
Let $B'\in\mathbb{R}^{n\times n}$ be an $\epsilon$-spectral-norm sparsifier of $A'$. 
If some row $B'_j$ has $s$ non-zeros then using~\eqref{eq:c_tail}, 
\ifaistats
\begin{align*}
    \|A'-B'\|_2
  \geq \|A'_j-B'_j\|_2
  & \geq \|a'_{\tail(s)}\|_2 \\
  & = \Omega(\sqrt{k}(s/k)^{-(\frac{1}{2}+\alpha)}) .
\end{align*}
\else
\[
\|A'-B'\|_2
  \geq \|A'_j-B'_j\|_2
   \geq \|a'_{\tail(s)}\|_2 
   = \Omega(\sqrt{k}(s/k)^{-(\frac{1}{2}+\alpha)}) .
  \]
\fi

By the error bound $\|A'-B'\|_2\leq \epsilon\|A'\|_2$,
we must have $s\geq \Omega(k(\epsilon/\alpha)^{-\frac{2}{1+2\alpha}})$,
which bounds the sparsity of every row and every column of $B'$.

We can set $\alpha=\log^{-1}\tfrac{1}{\epsilon}>\epsilon$. Note that this choice for $\alpha$ is in the range $[\log^{-1} \tfrac{n}{k},1]$, hence the construction hold. Now since $\tfrac{1}{1+2\alpha}\geq 1-2\alpha$, the lower bound on the sparsity of each row and each column of $B'$ is 
$k(\epsilon/\alpha)^{-\tfrac{2}{1+2\alpha}}\geq k(\epsilon/\alpha)^{-2+4\alpha}\geq \Omega(k\epsilon^{-2}\log^{-2}\tfrac{1}{\epsilon})$.
\end{proof}


\section{\ifaistats APPLICATION I: APPROXIMATE MATRIX MULTPLICATION \else Application I: Approximate Matrix Multiplication \fi}\label{sec:AMM}
In this section, we show how to use 
$\ell_1$ row/column sampling
for fast approximate matrix multiplication (AMM). 
Given matrices $A\in\mathbb{R}^{m\times n},B\in\mathbb{R}^{n\times p}$ and error parameter $\epsilon>0$, 
the goal is to compute a matrix $C \in\mathbb{R}^{m\times p}$ such that $\|AB-C\|\leq \epsilon \|A\|\cdot\|B\|$, where the norm is usually either the Frobenius-norm $\|\cdot\|_F$ or spectral-norm $\|\cdot\|_2$.
We provide the first results on AMM for numerically sparse matrices with respect to both norms.

\AMMspectralNorm*

\AMMfrobeniusNorm*

The proofs of these theorems combine Theorem \ref{thm:sparsifierForAMM} with previous results on numerical sparsity and with previous results on AMM.

\begin{lemma}[Lemma 4 of \cite{gupta2018exploiting}]\label{lem:vectorEstimation_GS18}
Given a vector $a\in\mathbb{R}^n$ and a parameter $\epsilon>0$, independently sampling $(\epsilon^{-2}\ns(a))$ entries according to the distribution $\{p_i=\frac{|a_i|}{\|a\|_1}\}_i$ and re-weighting the sampled coordinates by $\frac{1}{p_i}\cdot \frac{1}{\epsilon^{-2}\ns(a)}$, outputs
a $(\epsilon^{-2}\ns(a))$-sparse vector $a'\in\mathbb{R}^n$ satisfying $\E a'=a$ and $\E(\|a'\|_2^2)\leq(1+\epsilon^2)\|a\|_2^2$.
\end{lemma}

\subsection{Proof of Theorem \ref{thm:AMM_spectralNorm} (Spectral-Norm AMM)}

In order to prove Theorem \ref{thm:AMM_spectralNorm}, we will use a result from \cite{magen2011low}. Given matrices $A,B$, their product is $AB=\sum_i A^i B_i$. The algorithm in \cite{magen2011low} samples corresponding pairs of columns from $A$ and rows from $B$, hence the time it takes to compute an approximation of $AB$ depends on the sparsity of these rows and columns.
\begin{lemma}[Theorem 3.2 (ii) of \cite{magen2011low}.]\label{lem:AMM_of_MZ11_spectral}
There exists an algorithm that, given matrices $A\in\mathbb{R}^{m\times n},B\in\mathbb{R}^{n\times p}$, a parameter $0<\epsilon<1/2$ and constant factor estimates of $\|A\|_2$ and $\|B\|_2$, computes in time \ifaistats $O(\nnz(A)+\nnz(B))+\tilde{O}(\epsilon^{-2}\csp(A)\rsp(B)\sqrt{\sr(A)\sr(B)})$
\else
$$O\Big(\nnz(A)+\nnz(B)+\epsilon^{-2}\csp(A)\rsp(B)\sqrt{\sr(A)\sr(B)}\log\big(\epsilon^{-1}\sr(A)\sr(B)\big)\Big)$$ 
\fi
a matrix $C$ that satisfies \ifaistats $\mathbb{P}(\|C-A B\|_2\geq \epsilon \|A\|_2 \|B\|_2)\leq \frac{1}{\poly(\sr(A)\sr(B))}.$
\else
\[\mathbb{P}(\|C-A B\|_2\geq \epsilon \|A\|_2 \|B\|_2)\leq \frac{1}{\poly(\sr(A)\sr(B))}.\]
\fi
\end{lemma}

\ifcolt
\begin{proof}\textbf{of Theorem \ref{thm:AMM_spectralNorm}.}
\else
\begin{proof}[Proof of Theorem \ref{thm:AMM_spectralNorm}]
\fi
Given $\epsilon>0$, our algorithm is as follows. 
\begin{enumerate} \compactify
    \item Apply the algorithm in Theorem \ref{thm:sparsifierForAMM} on $A$ with parameter $\epsilon/4$ to compute a matrix $A'$ satisfying $\|A'-A\|_2\leq\tfrac{\epsilon}{4}\|A\|_2$ and $\csp(A')\leq O(\epsilon^{-2}\ns(A)\log (m+n))$, and apply it on $B$ with parameter $\epsilon/4$ to compute a matrix $B'$ satisfying $\|B'-B\|_2\leq \tfrac{\epsilon}{4}\|B\|_2$ and $\rsp(B')\leq O(\epsilon^{-2}\ns(B)\log (n+p))$. 
    \item Apply the algorithm in Lemma \ref{lem:AMM_of_MZ11_spectral} on $A',B'$ with parameter $\epsilon/4$ to produce a matrix $C$. Output $C$.
\end{enumerate}
\ifcolt
It holds that $\E\|A'\|_F^2\leq \big(1+O(\tfrac{\epsilon^2}{ \log(m+n)})\big) \|A\|_F^2$, since the sampling in Theorem \ref{thm:sparsifierForAMM} satisfies the conditions for Lemma \ref{lem:vectorEstimation_GS18}.
\else
The sampling in Theorem \ref{thm:sparsifierForAMM} satisfies the conditions for Lemma \ref{lem:vectorEstimation_GS18}, hence 
$\E\|A'\|_F^2\leq \big(1+O(\tfrac{\epsilon^2}{ \log(m+n)})\big) \|A\|_F^2$. 
\fi
Thus, with high probability, $\sr(A')\in (1\pm O(\epsilon))\sr(A)$, and similarly for $B'$.
Ignoring the $\nnz(\cdot)$ terms, the time it takes for the algorithm from Lemma \ref{lem:AMM_of_MZ11_spectral} on $A',B'$ is \ifaistats
$\tilde{O}(\epsilon^{-6}\ns(A)\ns(B)\sqrt{\sr(A)\sr(B)}),$
\else
$$O\Big(\epsilon^{-6}\ns(A)\ns(B)\log (m+n)\log (n+p)\sqrt{\sr(A)\sr(B)}\log\big(\epsilon^{-1}\sr(A)\sr(B)\big)\Big),$$
\fi
hence the stated overall running time. The output $C$ satisfies with high probability
\ifaistats 
\begin{align*}
    &\|AB-C\|_2 \leq \\
    & \leq \|(A-A')B\|_2 + \|(A'(B-B')\|_2 + \|A'B'-C\|_2 \\
    &\leq \|A\|_2\|B\|_2 (\tfrac{\epsilon}{4}  + \tfrac{\epsilon}{4}(1+\tfrac{\epsilon}{4}) + \tfrac{\epsilon}{4}(1+\tfrac{\epsilon}{4})^2) \leq \epsilon\|A\|_2\|B\|_2.
\end{align*}
\else
\begin{align*}
\|AB-C\|_2 &\leq \|(A-A')B\|_2 + \|(A'(B-B')\|_2 + \|A'B'-C\|_2 
    \\
    & \leq \tfrac{\epsilon}{4} \|A\|_2\|B\|_2 + \tfrac{\epsilon}{4}\|B\|_2(1+\tfrac{\epsilon}{4})\|A\|_2 + \tfrac{\epsilon}{4}(1+\tfrac{\epsilon}{4})^2 \|A\|_2\|B\|_2 \leq \epsilon\|A\|_2\|B\|_2.
\end{align*}
\fi
\end{proof}

\subsection{Proof of Theorem \ref{thm:AMM_frobeniusNorm} (Frobenius-Norm AMM)}

We provide a sampling lemma for estimating outer products in the Frobenius-norm.
\begin{lemma}\label{lem:approxOuterProduct}
There exists an algorithm that,
given vectors $a\in\mathbb{R}^n,b\in\mathbb{R}^m$ and parameter $0<\epsilon<1$, 
computes in time $O(\|a\|_0+\|b\|_0)$ vectors $a',b'\in\mathbb{R}^n$ with sparsity $\epsilon^{-2}\ns(a)$ and $\epsilon^{-2}\ns(b)$, respectively, satisfying $\E(a'b'^\top)=ab^\top$ and $\E\|a'b'^\top-ab^\top\|_F^2\leq\epsilon^2\|a\|_2^2\|b\|_2^2$.
\end{lemma}

\begin{proof}
Given $0<\epsilon<1$, our algorithm is as follows.
\begin{enumerate}
    \item Independently sample (with repetitions) $9\epsilon^{-2}\ns(a)$ entries from $a$ according to the distribution $\{p_i^{(a)}=\frac{|a_i|}{\|a\|_1}\}_i$ and $9\epsilon^{-2}\ns(b)$ entries from $b$ according to the distribution $\{p_i^{(b)}=\frac{|b_i|}{\|b\|_1}\}_i$.
    \item Re-weight the sampled entries of $a$ by $\frac{1}{p_i^{(a)}}\cdot \frac{1}{9\epsilon^{-2}\ns(a)}$ and similarly for $b$. Output the sampled vectors.
\end{enumerate}
Denote the sampled vectors $a'$ and $b'$. They satisfy the conditions of Lemma \ref{lem:vectorEstimation_GS18}, hence they satisfy $\E(a'b'^\top)=ab^\top$ and $\E(\|a'\|_2^2)\leq (1+\epsilon^2/3)\|a\|_2^2$ and similarly for $b'$. Thus, \ifaistats $\E\|a'b'^\top-ab^\top\|_F^2 =\E\|a'b'^\top\|_F^2-\|ab^\top\|_F^2 
=\E\|a'\|_2^2\|b'\|_2^2-\|a\|_2^2\|b\|_2^2\leq \epsilon^2\|a\|_2^2\|b\|_2^2.$ 
\else \ifcolt

$\E\|a'b'^\top-ab^\top\|_F^2=\E\|a'b'^\top\|_F^2-\|ab^\top\|_F^2
=\E\|a'\|_2^2\|b'\|_2^2-\|a\|_2^2\|b\|_2^2\leq \epsilon^2\|a\|_2^2\|b\|_2^2.$
\else
\[\E\|a'b'^\top-ab^\top\|_F^2=\E\|a'b'^\top\|_F^2-\|ab^\top\|_F^2
=\E\|a'\|_2^2\|b'\|_2^2-\|a\|_2^2\|b\|_2^2\leq \epsilon^2\|a\|_2^2\|b\|_2^2.\]
\fi
\end{proof}
In order to prove Theorem \ref{thm:AMM_frobeniusNorm}, we will use a result from \cite{drineas2006fast}. The algorithm in \cite{drineas2006fast} samples corresponding pairs of columns from $A$ and rows from $B$, hence the time it takes to compute an approximation of $AB$ depends on the sparsity of these rows and columns.
\begin{lemma}[Lemma 4 of \cite{drineas2006fast}]\label{lem:AMM_of_DKM06_frobenius}
There exists an algorithm that,
given matrices $A\in\mathbb{R}^{m\times n},B\in\mathbb{R}^{n\times p}$ and parameter $0<\epsilon<1$, 
computes in time $O(\nnz(A)+\nnz(B)+\epsilon^{-2}\csp(A)\rsp(B))$
a matrix $C\in\mathbb{R}^{m\times p}$ satisfying $\E\|AB-C\|_F\leq \epsilon\|A\|_F\|B\|_F$.
\end{lemma}

\ifcolt
\begin{proof}\textbf{of Theorem \ref{thm:AMM_frobeniusNorm}.}
\else
\begin{proof}[Proof of Theorem \ref{thm:AMM_frobeniusNorm}]
\fi
Let $0<\epsilon<1$. Recall that $AB=\sum_i A^i B_i$. Our algorithm is as follows. 
\begin{enumerate} \compactify
    \item Apply the algorithm in Lemma \ref{lem:approxOuterProduct} on each pair of vectors $A^i,B_i$ with parameter $\epsilon/3$ to obtain their sparse estimates $\hat{A}^i$ and $\hat{B}_i$.
    \item Arrange the column vectors $\{\hat{A}^i\}$ in a matrix $\hat{A}$ and the row vectors $\{\hat{B}_i\}$ in a matrix $\hat{B}$.
    \item Apply the algorithm in Lemma \ref{lem:AMM_of_DKM06_frobenius} on the matrices $\hat{A}$ and $\hat{B}$ with parameter $\epsilon/3$ to obtain their approximate product $C$. Output $C$.
\end{enumerate}
The sparsity of the columns of $\hat{A}$ is bounded by $\epsilon^{-2}\ns(A)$ and the sparsity of the rows of $\hat{B}$ is bounded by $\epsilon^{-2}\ns(B)$.
By the triangle inequality, Jensen inequality, Lemma \ref{lem:approxOuterProduct} and Cauchy-Schwarz inequality, 
\ifaistats
\begin{align*}
    &\E\|AB-\hat{A}\hat{B}\|_F =\E\|\sum_i A^i B_i - \hat{A}^i \hat{B}_i\|_F \\
    &\leq \sum_i\E\|A^i B_i - \hat{A}^i \hat{B}_i\|_F\\
    &\leq \sum_i\sqrt{\E\|A^i B_i - \hat{A}^i \hat{B}_i\|_F^2}\leq \tfrac{\epsilon}{3}\sum_i\|A^i\|_2 \|B_i\|_2\\
    &\leq \tfrac{\epsilon}{3}\|A\|_F\|B\|_F.
\end{align*}
\else
\begin{align*}
    \E\|AB-\hat{A}\hat{B}\|_F&=\E\|\sum_i A^i B_i - \hat{A}^i \hat{B}_i\|_F\leq \sum_i\E\|A^i B_i - \hat{A}^i \hat{B}_i\|_F\\
    &\leq \sum_i\sqrt{\E\|A^i B_i - \hat{A}^i \hat{B}_i\|_F^2}\leq \tfrac{\epsilon}{3}\sum_i\|A^i\|_2 \|B_i\|_2\leq \tfrac{\epsilon}{3}\|A\|_F\|B\|_F.
\end{align*}
\fi
Additionally, by Jensen's inequality and Lemma \ref{lem:vectorEstimation_GS18}, 
\ifaistats
\begin{align*}
  \E\|\hat{A}\|_F\leq \sqrt{\E\|\hat{A}\|_F^2}&\leq \sqrt{\sum_i (1+\tfrac{\epsilon^2}{9})\|A^i\|_2^2}\\
  &\leq (1+\tfrac{\epsilon}{3})\|A\|_F,  
\end{align*}
\else
\[
\E\|\hat{A}\|_F\leq \sqrt{\E\|\hat{A}\|_F^2}\leq \sqrt{\sum_i (1+\tfrac{\epsilon^2}{9})\|A^i\|_2^2}
  \leq (1+\tfrac{\epsilon}{3})\|A\|_F,  
  \]
\fi
and similarly for $\hat{B}$.
By the triangle inequality and Lemma \ref{lem:AMM_of_DKM06_frobenius}, \ifaistats 
\begin{align*}
    & \E\|C-AB\|_F\leq \E(\|C-\hat{A}\hat{B}\|_F + \|\hat{A}\hat{B}-AB\|_F) \\
    & \leq \tfrac{\epsilon}{3}(1+\tfrac{\epsilon}{3})^2\|A\|_F\|B\|_F + \tfrac{\epsilon}{3}\|A\|_F\|B\|_F \\
    & \leq \epsilon\|A\|_F\|B\|_F.   
\end{align*}
\else\ifcolt
\begin{align*}
    \E\|C-AB\|_F&\leq \E(\|C-\hat{A}\hat{B}\|_F + \|\hat{A}\hat{B}-AB\|_F)\\
    &\leq \tfrac{\epsilon}{3}(1+\tfrac{\epsilon}{3})^2\|A\|_F\|B\|_F + \tfrac{\epsilon}{3}\|A\|_F\|B\|_F\leq \epsilon\|A\|_F\|B\|_F.
\end{align*}
\else
\[\E\|C-AB\|_F\leq \E(\|C-\hat{A}\hat{B}\|_F + \|\hat{A}\hat{B}-AB\|_F)\leq \tfrac{\epsilon}{3}(1+\tfrac{\epsilon}{3})^2\|A\|_F\|B\|_F + \tfrac{\epsilon}{3}\|A\|_F\|B\|_F\leq \epsilon\|A\|_F\|B\|_F. \]
\fi
Except for the $\nnz(\cdot)$ terms, the time it takes to compute the last step is $O(\epsilon^{-6}\ns(A)\ns(B))$, and the claimed running time follows.
\end{proof}

\iffull

\section{Application II: Preconditioning for Ridge Regression}\label{sec:PreconRegProblem}
Often, problem-specific preconditioners are used to reduce the condition number of the problem, since the time it takes for iterative methods to converge depends on the condition number. Specifically, for a matrix $M\in\R^{n\times n}$ and a linear-system $Mx=b$, any \emph{invertible} matrix $P \in \R^{n \times n}$ has the property that the solution to the preconditioned linear-system $P^{-1}Mx=P^{-1} b$, is the same as that of the original problem. Using iterative methods to solve the preconditioned problem requires to apply $P^{-1}M$ to a vector in each iteration. In the case of ridge regression, $M=A^\top A +\lambda I$. Applying $(A^\top A +\lambda I)$ to a vector can be done in $O(\nnz(A))$ time, and applying $P^{-1}$ to a vector is equivalent to solving a linear-system in $P$, i.e. $\arg \min_x \|P x-y\|_2^2$ for some $y \in \R^n$. There is a trade-off between the number of iterations taken to converge for the preconditioned problem, and the time taken to (approximately) solve a linear-system in $P$.   
We show in this section how to use the sparsification scheme of Theorem \ref{thm:spectralSparsifier} to construct a preconditioner for ridge-regression, and couple it with an acceleration framework by \cite{frostig2015regularizing}.

\generalRowNormsRR*

Since the term $Ax$ is a linear combination of the columns of $A$, and the regularization term $\lambda\|x\|_2^2$ penalizes each coordinate of $x$ equally, in practice, the columns of $A$ are often pre-processed to have uniform norms before solving ridge-regression.
For this case, in section \ref{sec:RRUnifromNorms}, we show an improvement of roughly $(n/\ns(A))^{1/6}$ over Theorem \ref{thm:mainRidgeRegression}.

We start by showing that given a matrix $A\in\mathbb{R}^{m\times n}$ and parameter $\lambda>0$, if $P \in \R^{m \times n}$ is an $\epsilon$-spectral-norm sparsifier for $A$, for small enough $\epsilon$, the preconditioned problem has a constant condition number, hence requires only a constant number of iterations as described above. This was explored by \cite{gupta2018exploiting}, but they demanded $\epsilon$ to be $O(\frac{\lambda}{\|A\|_2^2})$, which is much smaller than necessary. In the next lemma we provide a tighter bound for $\epsilon$.


\begin{lemma}\label{lem:pIsGoodPreconditioner}
Given matrix $A\in\mathbb{R}^{m\times n}$, parameters $\lambda>0$ and $0<\epsilon'<\frac{1}{2}$, then
if a matrix $P\in\mathbb{R}^{m\times n}$ satisfies
$\|A-P\|_2<\epsilon \|A\|_2$ where $\epsilon = \frac{\sqrt{\lambda} \epsilon'}{\|A\|_2}$, then
\[
(1-2\epsilon')(A^\top A+\lambda I) \preceq P^\top P+\lambda I\preceq (1+2\epsilon')(A^\top A+\lambda I).
\]
\end{lemma}
Setting $\epsilon' = 1/4$ yields that all the eigenvalues of $(P^\top P + \lambda I)^{-1} (A^\top A + \lambda I)$ are in the range $[\frac{2}{3}, 2]$. 
Using our sampling scheme in Theorem \ref{thm:spectralSparsifier} with parameter $\epsilon$ as described here, denoting its output as $P$, provides a preconditioner for ridge regression with constant condition number. Hence solving this preconditioned problem, i.e, the linear-system $(P^\top P +\lambda I)^{-1}(A^\top A +\lambda I)x=(P^\top P +\lambda I)^{-1}b$ for some vector $b\in\R^n$, with any iterative method, takes $O_\epsilon(\nnz(A)+T_P^\lambda)$ time, where $T_P^\lambda$ is the time it takes to compute an approximate solution to $\arg \min_x \|(P^\top P +\lambda I) x-y\|_2^2$ for some vector $y\in\R^n$.

\ifcolt
\begin{proof}\textbf{of Lemma \ref{lem:pIsGoodPreconditioner}.}
For any $x\in\mathbb{R}^n$, by the Triangle inequality,
\[
\|P x\|_2\leq \|Ax\|_2 + \|(P - A) x\|_2 \leq  \|Ax\|_2 +  \sqrt{\lambda} \epsilon' \|x\|_2.
\]
By squaring both sides and applying the AM-GM inequality,
\begin{align*}
\|P x\|_2^2 & \leq
\|A x\|_2^2 +  \lambda \epsilon'^2 \|x\|_2^2 + 2\|Ax\|_2\sqrt{\lambda} \epsilon'\|x\|_2 \\
& \leq \|Ax\|_2^2 +  \lambda \epsilon'^2 \|x\|_2^2 + \epsilon'\left(\|Ax\|_2^2 + \lambda \|x\|_2^2\right)  \\
& = (1+\epsilon') \|Ax\|_2^2 + \lambda \epsilon'(1+\epsilon') \|x\|_2^2,
\end{align*}
and since $\epsilon'<1$,
\[
\|P x\|_2^2 +\lambda \|x\|_2^2
\leq (1+\epsilon') \|Ax\|_2^2 + \lambda (1+2\epsilon') \|x\|_2^2.
\]
Hence $P^\top P+\lambda I\preceq (1+2\epsilon')(A^\top A+\lambda I)$. 

Similarly, we get $\|P x\|_2\geq  \|Ax\|_2 -  \sqrt{\lambda} \epsilon' \|x\|_2$, thus $\|P x\|_2^2 \geq (1-\epsilon')\|Ax\|_2^2 - \lambda\epsilon'(1-\epsilon')\|x\|_2^2$ and
\[
\|Px\|_2^2 +\lambda \|x\|_2^2 \geq (1 - \epsilon')\|Ax\|_2^2 + \lambda (1-2\epsilon')\|x\|_2^2.
\]

\end{proof}
\else
\begin{proof}[Proof of Lemma \ref{lem:pIsGoodPreconditioner}]
For any $x\in\mathbb{R}^n$, by the Triangle inequality,
\[
\|P x\|_2\leq \|Ax\|_2 + \|(P - A) x\|_2 \leq  \|Ax\|_2 +  \sqrt{\lambda} \epsilon' \|x\|_2.
\]
By squaring both sides and applying the AM-GM inequality,
\begin{align*}
\|P x\|_2^2 & \leq
\|A x\|_2^2 +  \lambda \epsilon'^2 \|x\|_2^2 + 2\|Ax\|_2\sqrt{\lambda} \epsilon'\|x\|_2 \\
& \leq \|Ax\|_2^2 +  \lambda \epsilon'^2 \|x\|_2^2 + \epsilon'\left(\|Ax\|_2^2 + \lambda \|x\|_2^2\right)  \\
& = (1+\epsilon') \|Ax\|_2^2 + \lambda \epsilon'(1+\epsilon') \|x\|_2^2,
\end{align*}
and since $\epsilon'<1$,
\[
\|P x\|_2^2 +\lambda \|x\|_2^2
\leq (1+\epsilon') \|Ax\|_2^2 + \lambda (1+2\epsilon') \|x\|_2^2.
\]
Hence $P^\top P+\lambda I\preceq (1+2\epsilon')(A^\top A+\lambda I)$. 

Similarly, we get $\|P x\|_2\geq  \|Ax\|_2 -  \sqrt{\lambda} \epsilon' \|x\|_2$, thus $\|P x\|_2^2 \geq (1-\epsilon')\|Ax\|_2^2 - \lambda\epsilon'(1-\epsilon')\|x\|_2^2$ and
\[
\|Px\|_2^2 +\lambda \|x\|_2^2 \geq (1 - \epsilon')\|Ax\|_2^2 + \lambda (1-2\epsilon')\|x\|_2^2.
\]

\end{proof}
\fi

\subsection{Proof of Theorem \ref{thm:mainRidgeRegression}}
Solving the linear-system in $P^\top P +\lambda I$ can be done by the Conjugate Gradient (CG) method, and can be 
accelerated by the framework of \cite{frostig2015regularizing}, that, given an algorithm to compute an approximate solution to an Empirical Risk Minimization (ERM) problem, uses the algorithm to provide acceleration in a black-box manner.
We restate the guarantees for these algorithms below. 

\begin{fact}
\label{fact:CGRuntime}
For a matrix $M \in \R^{m \times n}$, vector $y \in \R^m$ and parameters $\epsilon, \lambda > 0$, the Conjugate Gradient algorithm returns an $\epsilon$-approximate solution to 
$\min_x \|Mx-y\|_2^2 +\lambda \|x\|_2^2$
in time 
\ifcolt
$$T^\lambda_{\mathsf{CG}}(M, \epsilon)\eqdef O(\nnz(M)\sqrt{\kappa_\lambda(M)}\log(\tfrac{1}{\epsilon})).$$
\else
$O(\nnz(M)\sqrt{\kappa_\lambda(M)}\log(\tfrac{1}{\epsilon}))$, which we will denote by $T^\lambda_{\mathsf{CG}}(M, \epsilon)$.
\fi
\end{fact}

\begin{lemma}[Acceleration. Theorem 1.1 of \cite{frostig2015regularizing}] 
\label{thm:acceleration}
Let $f:\mathbb{R}^n\rightarrow \mathbb{R}$ be a $\lambda$ strongly convex function and for all $x_0\in\mathbb{R}^n, c>1, \lambda'>0$, let $f_{min}=\min_{x\in\mathbb{R}^n}(f(x)+\frac{\lambda'}{2}\|x-x_0\|_2^2)$, assume we can compute $x_c\in\mathbb{R}^n$ in time $T_c$ such that
\[
\E(f(x_c))-f_{min}\leq
\tfrac{1}{c}(f(x_0)-f_{min}),
\]
then, given any $x_0\in\mathbb{R}^n,c>1,\lambda'\geq 2\lambda$, we can compute $x_1$ such that
\[
\E(f(x_1))-\min_x(f(x))\leq \tfrac{1}{c}\big(f(x_0)-\min_x(f(x))\big)
\]
in time $O\Big(T_{4(\frac{2\lambda'+\lambda}{\lambda})^{1.5}}\sqrt{\frac{\lambda'}{\lambda}}\log c\Big)$.
\end{lemma}
The measure of error in the above theorem coincides with the definition we gave for $\epsilon$-approximation to ridge regression, since if $f(x)=\|Ax-b\|_2^2 + \lambda \|x\|_2^2$ and $x^*=\argmin_x f(X)$ then for any $x\in\R^n$,
$(x-x^*)^T(A^TA+\lambda I)(x-x^*)=2(f(x)-f(x^*))$. For a proof, see for example \cite[Fact 39]{musco2017spectrum}.

Note that the term $\frac{\lambda'}{2}\|x-x_0\|_2^2$ is not exactly of the same shape as the ridge term $\lambda'\|x\|_2^2$, but since
\[
\|Ax-b\|_2^2 + \lambda\|x\|_2^2
+\lambda'\|x-x_0\|_2^2=
\|Ax-b\|_2^2+(\lambda+\lambda')\|x\|_2^2-2\lambda'x_0^\top x+\lambda'\|x_0\|_2^2,
\]
solving 
$\min_x(\|Ax-b\|_2^2 + \lambda\|x\|_2^2+\lambda'\|x-x_0\|_2^2)$ 
is at most as hard as solving ridge regression with vector $A^Tb+\lambda'x_0$ and parameter $\lambda+\lambda'$.
We are now ready to prove the result for preconditioned ridge-regression using our sparsifier as a preconditioner. 

\ifcolt
\begin{proof}\textbf{of Theorem \ref{thm:mainRidgeRegression}.}
We first explain how to compute an approximate solution for ridge regression with parameter $\lambda>0$ and then apply the acceleration framework of Lemma \ref{thm:acceleration} as a black-box.

Apply the sparsification scheme of Theorem \ref{thm:spectralSparsifier} on $A$ with parameter $\epsilon = \frac{\sqrt{\lambda}}{4\|A\|_2}$ as specified in Lemma \ref{lem:pIsGoodPreconditioner} and denote its output by $P$.
Solve the preconditioned linear-system $(P^\top P+\lambda I)^{-1}(A^\top A +\lambda I)x=(P^\top P+\lambda I)^{-1}b$ by any iterative method. 
As was described earlier, this takes $O_\epsilon(\nnz(A) + T_P^\lambda)$ time.
Use Conjugate gradients to solve each linear-system in $P^\top P +\lambda I$.
It takes $O_\epsilon(\sqrt{\kappa_\lambda(P^\top P)}\nnz(P))$ time. Since $\|P\|_2\in (1\pm \epsilon)\|A\|_2$ and $\kappa_\lambda(P^\top P)=\frac{\|P\|_2^2}{\lambda}$, by Theorem \ref{thm:spectralSparsifier},
\ifcolt
\begin{align*}
    T^\lambda_{\mathsf{CG}}(P,\epsilon)&=O_\epsilon\Big(\nnz(P)\sqrt{\kappa_\lambda(P^\top P)}\Big)\\
    &=O_\epsilon\Big(\frac{\|A\|_2^3}{\lambda^{1.5}}\ns(A) \sr(A)\log n+\frac{\|A\|_2^2}{\lambda}\sqrt{\ns(A) n\cdot \sr(A)}\log n\Big).
\end{align*}
\else
\[
T^\lambda_{\mathsf{CG}}(P,\epsilon)=O_\epsilon\Big(\nnz(P)\sqrt{\kappa_\lambda(P^\top P)}\Big)
=O_\epsilon\Big(\frac{\|A\|_2^3}{\lambda^{1.5}}\ns(A) \sr(A)\log n+\frac{\|A\|_2^2}{\lambda}\sqrt{\ns(A) n\cdot \sr(A)}\log n\Big).
\]
\fi
Applying the acceleration framework (Lemma \ref{thm:acceleration}) yields a running time of
\[
\tilde{O}\bigg(\Big(\nnz(A)+\frac{\|A\|_2^3}{\lambda'^{1.5}}\ns(A) \sr(A)+\frac{\|A\|_2^2}{\lambda'}\sqrt{\ns(A) n\cdot \sr(A)}\Big)\sqrt{\frac{\lambda'}{\lambda}}\bigg).
\]
Set $\lambda'=\|A\|_2^2\big(\frac{\ns(A) \sr(A)}{\nnz(A)}\big)^{2/3}$. If $n<\ns(A) \sr(A)\frac{\|A\|_2^2}{\lambda'}=(\ns(A) \sr(A))^{1/3}(\nnz(A))^{2/3}$, which is a reasonable assumption in many cases (for example, if $\nnz(A)>n^{3/2}$), then this choice for $\lambda'$ balances the two major terms, resulting in the stated running time.
\end{proof}
\else
\begin{proof}[Proof of Theorem \ref{thm:mainRidgeRegression}]
We first explain how to compute an approximate solution for ridge regression with parameter $\lambda>0$ and then apply the acceleration framework of Lemma \ref{thm:acceleration} as a black-box.

Apply the sparsification scheme of Theorem \ref{thm:spectralSparsifier} on $A$ with parameter $\epsilon = \frac{\sqrt{\lambda}}{4\|A\|_2}$ as specified in Lemma \ref{lem:pIsGoodPreconditioner} and denote its output by $P$.
Solve the preconditioned linear-system $(P^\top P+\lambda I)^{-1}(A^\top A +\lambda I)x=(P^\top P+\lambda I)^{-1}b$ by any iterative method. 
As was described earlier, this takes $O_\epsilon(\nnz(A) + T_P^\lambda)$ time.
Use Conjugate gradients to solve each linear-system in $P^\top P +\lambda I$.
It takes $O_\epsilon(\sqrt{\kappa_\lambda(P^\top P)}\nnz(P))$ time. Since $\|P\|_2\in (1\pm \epsilon)\|A\|_2$ and $\kappa_\lambda(P^\top P)=\frac{\|P\|_2^2}{\lambda}$, by Theorem \ref{thm:spectralSparsifier},
\[
T^\lambda_{\mathsf{CG}}(P,\epsilon)=O_\epsilon\Big(\nnz(P)\sqrt{\kappa_\lambda(P^\top P)}\Big)
=O_\epsilon\Big(\frac{\|A\|_2^3}{\lambda^{1.5}}\ns(A) \sr(A)\log n+\frac{\|A\|_2^2}{\lambda}\sqrt{\ns(A) n\cdot \sr(A)}\log n\Big).
\]
Applying the acceleration framework (Lemma \ref{thm:acceleration}) yields a running time of
\[
\tilde{O}\bigg(\Big(\nnz(A)+\frac{\|A\|_2^3}{\lambda'^{1.5}}\ns(A) \sr(A)+\frac{\|A\|_2^2}{\lambda'}\sqrt{\ns(A) n\cdot \sr(A)}\Big)\sqrt{\frac{\lambda'}{\lambda}}\bigg).
\]
Set $\lambda'=\|A\|_2^2\big(\frac{\ns(A) \sr(A)}{\nnz(A)}\big)^{2/3}$. If $n<\ns(A) \sr(A)\frac{\|A\|_2^2}{\lambda'}=(\ns(A) \sr(A))^{1/3}(\nnz(A))^{2/3}$, which is a reasonable assumption in many cases (for example, if $\nnz(A)>n^{3/2}$), then this choice for $\lambda'$ balances the two major terms, resulting in the stated running time.
\end{proof}
\fi

\subsection{Faster Algorithm for Inputs with Uniform Row Norms}\label{sec:RRUnifromNorms}

The best running time, to our knowledge, 
for the ridge-regression problem on sparse matrices in general is using 
Stochastic Variance Reduced Gradient Descent (SVRG), originally introduced by \cite{johnson2013accelerating}, coupled with the acceleration framework of \cite{frostig2015regularizing}. We utilize this method for solving the linear-system for $P^\top P +\lambda I$, where $P$ is the preconditioner. 
This method is fastest if the norms of the rows/columns of the input matrix $A$ are uniform. 
We show the following theorem for solving ridge-regression on numerically sparse matrices with uniform row/column norms.

\begin{theorem}\label{thm:uniformRowNormsRR}
There exists an algorithm that,
given a matrix $A\in \R^{m \times n}$ having uniform rows norms or uniform columns norms, a vector $x_0\in R^n$ and parameters $\lambda>0, \epsilon>0$, 
computes an $\epsilon$-approximate solution to the ridge regression problem
in expected time 
\[
O_\epsilon(\nnz(A)) + \tilde{O}_\epsilon\Big(\nnz(A)^{2/3} \sqrt{\sr(A)}\ns(A)^{1/3}  n^{-1/6} \sqrt{\kappa_\lambda(A^\top A)}\Big).
\]
\end{theorem}

Note that 
$(A^\top A +\lambda I)^{-1}A^\top = A^\top (A A^\top +\lambda I)^{-1}$.
Hence, for any vector $v$, one can compute an $\epsilon$-approximation for $(A^\top A +\lambda I)^{-1}A^\top v$ 
in time $O(\nnz(A))+T^\lambda (A^\top,\epsilon)$. This doesn't change the condition number of the problem, i.e, $\kappa_\lambda(A^\top A)=\kappa_\lambda (AA^\top)$. Hence we only analyze the case where $A$ is pre-processed such that the norms of the rows are uniform.



We provide a theorem from \cite{musco2017spectrum} that summarizes the running time of accelerated-SVRG. 

\begin{lemma}[Theorem 49 of \cite{musco2017spectrum}]\label{thm:bestSparseRidgeRuntime}
For a matrix $M$, vector $y \in \R^n$ and $\lambda, \epsilon > 0$, there exists an algorithm that computes with high probability an $\epsilon$-approximate solution to $\min_x \|Mx-y\|_2^2+\lambda\|x\|_2^2$ in time $T^\lambda(M, \epsilon)$ such that 
\begin{align*}
T^\lambda(M, \epsilon) &\leq O_\epsilon (\nnz(M)) + \tilde{O}_\epsilon \bigg(\sqrt{\nnz(M) \cdot \frac{\|M\|_F^2}{\lambda} \cdot \rsp(M) } \bigg).
\end{align*}
\end{lemma}

Before we prove Theorem \ref{thm:uniformRowNormsRR}, note the following properties of the sampling in Theorem \ref{thm:spectralSparsifier}.

\begin{lemma}\label{lem:rowNormBound}
Given a matrix $A\in\R^{m\times n}$, parameter $\epsilon>0$ and a random matrix $P\in\R^{m\times n}$ satisfying $\|P-A\|_2\leq \epsilon\|A\|_2$ and $\E P =A$, then the expected $\ell_2$-norm of the $i$-th row and of the $j$-th column of $P$ are bounded as
\[
\E\|P_i\|_2^2\leq \|A_i\|_2^2 + \epsilon^2\|A\|_2^2,
\]
\[
\E\|P^j\|_2^2\leq \|A^j\|_2^2 + \epsilon^2\|A\|_2^2.
\]
\end{lemma}
\begin{proof} 
By properties of the spectral-norm, $\|P_i-A_i\|_2\leq \|P-A\|_2\leq \epsilon\|A\|_2$. 
Squaring this and taking the expectation yields
$\E(\|P_i\|_2^2) - \|A_i\|_2^2 \leq \epsilon^2\|A\|_2^2$
as desired. The same holds for the columns. One can similarly get an high probability statement.
\end{proof}
Summing over all the rows or columns yields an immediate corollary,
\begin{corollary}\label{cor:frobeniusBound}
The expected Frobenius-norm of $P$ is bounded as
\ifcolt
$$\E\|P\|_F^2\leq \|A\|_F^2 + \epsilon^2\min(n,m)\|A\|_2^2.$$
\else
$\E\|P\|_F^2\leq \|A\|_F^2 + \epsilon^2\min(n,m)\|A\|_2^2.$
\fi
\end{corollary}

We are now ready to show the result for ridge-regression in the case that the norms of the rows of the input matrix $A$ are uniform.

\ifcolt
\begin{proof}\textbf{of Theorem \ref{thm:uniformRowNormsRR}.}
We first explain how to compute an approximate solution for ridge regression with parameter $\lambda>0$ and then apply the acceleration framework of Lemma \ref{thm:acceleration} as a black-box.

Apply the sparsification scheme of Theorem \ref{thm:spectralSparsifier} on $A$ with parameter $\epsilon = \frac{\sqrt{\lambda}}{4\|A\|_2}$ as specified in Lemma \ref{lem:pIsGoodPreconditioner} and denote its output by $P$.
Solve the preconditioned linear-system $(P^\top P+\lambda I)^{-1}(A^\top A +\lambda I)x=(P^\top P+\lambda I)^{-1}b$ by any iterative method. 
As was described earlier, this takes $O_\epsilon(\nnz(A) + T_P^\lambda)$ time.
Use Accelerated-SVRG (Lemma \ref{thm:bestSparseRidgeRuntime}) to solve each linear-system in $P^\top P +\lambda I$.

The bulk of the running time of the Accelerated-SVRG method is in applying vector-vector multiplication in each iteration, where one of the vectors is a row of $P$. The number of iterations have dependence on $\sr(P)$, which by Corollary \ref{cor:frobeniusBound} is bounded by $O(\sr(A)+\frac{n}{\kappa_\lambda})$. The running time of each iteration is usually bounded by the maximum row sparsity, i.e, $\rsp(P)$. Instead, we can bound the expected running time with the expected row sparsity, denote as $s^*(P)$. The distribution for sampling each row is $p_i=\frac{\|P_i\|_2^2}{\|P\|_F^2}$ ~\citep{musco2017spectrum}. Hence, the expected running time will depend on $\sum_i p_i\|P_i\|_0$ instead of $\rsp(P)$. 
By Lemma \ref{lem:rowNormBound} and the assumption that the norms of the rows of $A$ are uniform,
\begin{equation}\label{eq:expectedSparsityBound}
s^*(P) = \sum_i p_i\|P_i\|_0 \leq \sum_i \frac{\|A_i\|_2^2 + \lambda}{\|P\|_F^2}\|P_i\|_0 \leq \nnz(P)\Big(\frac{1}{n} + \frac{\lambda}{\|P\|_F^2}\Big).    
\end{equation}

Now, by Lemma \ref{thm:bestSparseRidgeRuntime}, equation \ref{eq:expectedSparsityBound} and corollary \ref{cor:frobeniusBound},  
\begin{align*}
 T^\lambda(\textsf{P}, \epsilon) &\leq O_\epsilon \Big(\nnz(P) +  \sqrt{\nnz(P) s^*(P)\sr(P) \cdot \kappa_\lambda(P^\top P) } \Big)\\
 &\leq O_\epsilon \bigg(\nnz(P)+ \nnz(P) \sqrt{\frac{\sr(P) \cdot \kappa_\lambda(P^\top P)}{n} +1} \bigg)\\ 
 &\leq O_\epsilon \bigg( \nnz(P) + \nnz(P) \sqrt{\frac{\sr(A) \cdot \kappa_\lambda(A^\top A)}{n}} \bigg)\\ 
 &\leq O_\epsilon(\nnz(P)) + \tilde{O}_\epsilon\Big( \frac{\kappa_\lambda(A^\top A)^{3/2}  \ns(A) \sr(A)^{3/2}}{\sqrt{n}}\Big).
\end{align*}

The last inequality is by plugging in $\nnz(P)$ for the second term. Applying the acceleration framework (Lemma \ref{thm:acceleration}) to the preconditioned problem (i.e, $P$ is a $\frac{c}{\sqrt{\kappa_{\lambda'}(A\top A)}}$-spectral-norm sparsifier of $A$), yields running time of
\[
\tilde{O}_\epsilon\bigg(\nnz(A) + \Big( \nnz(A) + \frac{\kappa_{\lambda'}(A^\top A)^{3/2}  \ns(A) \sr(A)^{3/2}}{\sqrt{n}}\Big)\sqrt{\frac{\lambda'}{\lambda}}\bigg)
\]
Setting $\lambda' = \frac{\|A\|_2^2 \ns(A)^{2/3} \sr(A)}{n^{1/3}\nnz(A)^{2/3}}$ results in the stated running time.
\end{proof}
\else
\begin{proof}[Proof of Theorem \ref{thm:uniformRowNormsRR}]
We first explain how to compute an approximate solution for ridge regression with parameter $\lambda>0$ and then apply the acceleration framework of Lemma \ref{thm:acceleration} as a black-box.

Apply the sparsification scheme of Theorem \ref{thm:spectralSparsifier} on $A$ with parameter $\epsilon = \frac{\sqrt{\lambda}}{4\|A\|_2}$ as specified in Lemma \ref{lem:pIsGoodPreconditioner} and denote its output by $P$.
Solve the preconditioned linear-system $(P^\top P+\lambda I)^{-1}(A^\top A +\lambda I)x=(P^\top P+\lambda I)^{-1}b$ by any iterative method. 
As was described earlier, this takes $O_\epsilon(\nnz(A) + T_P^\lambda)$ time.
Use Accelerated-SVRG (Lemma \ref{thm:bestSparseRidgeRuntime}) to solve each linear-system in $P^\top P +\lambda I$.

The bulk of the running time of the Accelerated-SVRG method is in applying vector-vector multiplication in each iteration, where one of the vectors is a row of $P$. The number of iterations have dependence on $\sr(P)$, which by Corollary \ref{cor:frobeniusBound} is bounded by $O(\sr(A)+\frac{n}{\kappa_\lambda})$. The running time of each iteration is usually bounded by the maximum row sparsity, i.e, $\rsp(P)$. Instead, we can bound the expected running time with the expected row sparsity, denote as $s^*(P)$. The distribution for sampling each row is $p_i=\frac{\|P_i\|_2^2}{\|P\|_F^2}$ ~\citep{musco2017spectrum}. Hence, the expected running time will depend on $\sum_i p_i\|P_i\|_0$ instead of $\rsp(P)$. 
By Lemma \ref{lem:rowNormBound} and the assumption that the norms of the rows of $A$ are uniform,
\begin{equation}\label{eq:expectedSparsityBound}
s^*(P) = \sum_i p_i\|P_i\|_0 \leq \sum_i \frac{\|A_i\|_2^2 + \lambda}{\|P\|_F^2}\|P_i\|_0 \leq \nnz(P)\Big(\frac{1}{n} + \frac{\lambda}{\|P\|_F^2}\Big)    
\end{equation}

Now, by Lemma \ref{thm:bestSparseRidgeRuntime}, equation \ref{eq:expectedSparsityBound} and corollary \ref{cor:frobeniusBound},  
\begin{align*}
 T^\lambda(\textsf{P}, \epsilon) &\leq O_\epsilon \Big(\nnz(P) +  \sqrt{\nnz(P) s^*(P)\sr(P) \cdot \kappa_\lambda(P^\top P) } \Big)\\
 &\leq O_\epsilon \bigg(\nnz(P)+ \nnz(P) \sqrt{\frac{\sr(P) \cdot \kappa_\lambda(P^\top P)}{n} +1} \bigg)\\ 
 &\leq O_\epsilon \bigg( \nnz(P) + \nnz(P) \sqrt{\frac{\sr(A) \cdot \kappa_\lambda(A^\top A)}{n}} \bigg)\\ 
 &\leq O_\epsilon(\nnz(P)) + \tilde{O}_\epsilon\Big( \frac{\kappa_\lambda(A^\top A)^{3/2}  \ns(A) \sr(A)^{3/2}}{\sqrt{n}}\Big).
\end{align*}

The last inequality is by plugging in $\nnz(P)$ for the second term. Applying the acceleration framework (Lemma \ref{thm:acceleration}) to the preconditioned problem (i.e, $P$ is a $\frac{c}{\sqrt{\kappa_{\lambda'}(A\top A)}}$-spectral-norm sparsifier of $A$), yields running time of
\[
\tilde{O}_\epsilon\bigg(\nnz(A) + \Big( \nnz(A) + \frac{\kappa_{\lambda'}(A^\top A)^{3/2}  \ns(A) \sr(A)^{3/2}}{\sqrt{n}}\Big)\sqrt{\frac{\lambda'}{\lambda}}\bigg)
\]
Setting $\lambda' = \frac{\|A\|_2^2 \ns(A)^{2/3} \sr(A)}{n^{1/3}\nnz(A)^{2/3}}$ results in the stated running time.
\end{proof}
\fi

\fi

\bibliography{references.bib}

\begin{thebibliography}{28}
\providecommand{\natexlab}[1]{#1}
\providecommand{\url}[1]{\texttt{#1}}
\expandafter\ifx\csname urlstyle\endcsname\relax
  \providecommand{\doi}[1]{doi: #1}\else
  \providecommand{\doi}{doi: \begingroup \urlstyle{rm}\Url}\fi

\bibitem[Achlioptas and McSherry(2007)]{achlioptas2007fast}
D.~Achlioptas and F.~McSherry.
\newblock Fast computation of low-rank matrix approximations.
\newblock \emph{Journal of the ACM (JACM)}, 54\penalty0 (2):\penalty0 9--es,
  2007.

\bibitem[Achlioptas et~al.(2013)Achlioptas, Karnin, and
  Liberty]{achlioptas2013near}
D.~Achlioptas, Z.~S. Karnin, and E.~Liberty.
\newblock Near-optimal entrywise sampling for data matrices.
\newblock In \emph{Advances in Neural Information Processing Systems}, pages
  1565--1573, 2013.

\bibitem[Arora et~al.(2005)Arora, Hazan, and Kale]{arora2005fast}
S.~Arora, E.~Hazan, and S.~Kale.
\newblock Fast algorithms for approximate semidefinite programming using the
  multiplicative weights update method.
\newblock In \emph{46th Annual IEEE Symposium on Foundations of Computer
  Science (FOCS'05)}, pages 339--348. IEEE, 2005.

\bibitem[Arora et~al.(2006)Arora, Hazan, and Kale]{arora2006fast}
S.~Arora, E.~Hazan, and S.~Kale.
\newblock A fast random sampling algorithm for sparsifying matrices.
\newblock In \emph{Approximation, Randomization, and Combinatorial
  Optimization. Algorithms and Techniques}, pages 272--279. Springer, 2006.

\bibitem[Carmon et~al.(2020)Carmon, Jin, Sidford, and
  Tian]{carmon2020coordinate}
Y.~Carmon, Y.~Jin, A.~Sidford, and K.~Tian.
\newblock Coordinate methods for matrix games.
\newblock In \emph{61st Annual {IEEE} Symposium on Foundations of Computer
  Science, {FOCS}}, pages 283--293. {IEEE}, 2020.

\bibitem[Clarkson and Woodruff(2009)]{clarkson2009numerical}
K.~L. Clarkson and D.~P. Woodruff.
\newblock Numerical linear algebra in the streaming model.
\newblock In \emph{Proceedings of the 41st Annual ACM SIGACT Symposium on
  Theory of Computing}, pages 205--214, 2009.

\bibitem[Cohen et~al.(2016)Cohen, Nelson, and Woodruff]{cohen2015optimal}
M.~B. Cohen, J.~Nelson, and D.~P. Woodruff.
\newblock Optimal approximate matrix product in terms of stable rank.
\newblock In \emph{43rd International Colloquium on Automata, Languages, and
  Programming (ICALP 2016)}. Schloss Dagstuhl-Leibniz-Zentrum fuer Informatik,
  2016.

\bibitem[d'Aspremont(2011)]{d2011semidefsubsampling}
A.~d'Aspremont.
\newblock Subsampling algorithms for semidefinite programming.
\newblock \emph{Stochastic Systems}, 1\penalty0 (2):\penalty0 274--305, 2011.

\bibitem[Drineas and Zouzias(2011)]{drineas2011note}
P.~Drineas and A.~Zouzias.
\newblock A note on element-wise matrix sparsification via a matrix-valued
  {B}ernstein inequality.
\newblock \emph{Information Processing Letters}, 111\penalty0 (8):\penalty0
  385--389, 2011.

\bibitem[Drineas et~al.(2006)Drineas, Kannan, and Mahoney]{drineas2006fast}
P.~Drineas, R.~Kannan, and M.~W. Mahoney.
\newblock Fast {M}onte {C}arlo algorithms for matrices {I}: Approximating
  matrix multiplication.
\newblock \emph{SIAM Journal on Computing}, 36\penalty0 (1):\penalty0 132--157,
  2006.

\bibitem[Frieze et~al.(2004)Frieze, Kannan, and Vempala]{frieze2004fast}
A.~Frieze, R.~Kannan, and S.~Vempala.
\newblock Fast {M}onte-{C}arlo algorithms for finding low-rank approximations.
\newblock \emph{Journal of the ACM (JACM)}, 51\penalty0 (6):\penalty0
  1025--1041, 2004.

\bibitem[Frostig et~al.(2015)Frostig, Ge, Kakade, and
  Sidford]{frostig2015regularizing}
R.~Frostig, R.~Ge, S.~Kakade, and A.~Sidford.
\newblock Un-regularizing: approximate proximal point and faster stochastic
  algorithms for empirical risk minimization.
\newblock In \emph{International Conference on Machine Learning}, pages
  2540--2548, 2015.

\bibitem[Ghashami et~al.(2016)Ghashami, Liberty, and
  Phillips]{ghashami2016efficient}
M.~Ghashami, E.~Liberty, and J.~M. Phillips.
\newblock Efficient frequent directions algorithm for sparse matrices.
\newblock In \emph{Proceedings of the 22nd ACM SIGKDD International Conference
  on Knowledge Discovery and Data Mining}, pages 845--854, 2016.

\bibitem[Gittens and Tropp(2009)]{gittens2009error}
A.~Gittens and J.~A. Tropp.
\newblock Error bounds for random matrix approximation schemes.
\newblock \emph{arXiv preprint arXiv:0911.4108}, 2009.

\bibitem[Gupta and Sidford(2018)]{gupta2018exploiting}
N.~Gupta and A.~Sidford.
\newblock Exploiting numerical sparsity for efficient learning: faster
  eigenvector computation and regression.
\newblock In \emph{Advances in Neural Information Processing Systems}, pages
  5269--5278, 2018.

\bibitem[Hoyer(2004)]{hoyer2004non}
P.~O. Hoyer.
\newblock Non-negative matrix factorization with sparseness constraints.
\newblock \emph{Journal of Machine Learning Research}, 5\penalty0 (9), 2004.

\bibitem[Huang(2019)]{pmlr-v80-huang18a}
Z.~Huang.
\newblock Near optimal frequent directions for sketching dense and sparse
  matrices.
\newblock \emph{Journal of Machine Learning Research}, 20\penalty0
  (56):\penalty0 1--23, 2019.

\bibitem[Hurley and Rickard(2009)]{hurley2009comparing}
N.~Hurley and S.~Rickard.
\newblock Comparing measures of sparsity.
\newblock \emph{IEEE Transactions on Information Theory}, 55\penalty0
  (10):\penalty0 4723--4741, 2009.

\bibitem[Johnson and Zhang(2013)]{johnson2013accelerating}
R.~Johnson and T.~Zhang.
\newblock Accelerating stochastic gradient descent using predictive variance
  reduction.
\newblock In \emph{Advances in Neural Information Processing Systems}, pages
  315--323, 2013.

\bibitem[Kundu and Drineas(2014)]{kundu2014note}
A.~Kundu and P.~Drineas.
\newblock A note on randomized element-wise matrix sparsification.
\newblock \emph{arXiv preprint arXiv:1404.0320}, 2014.

\bibitem[Kundu et~al.(2017)Kundu, Drineas, and
  Magdon-Ismail]{kundu2017recovering}
A.~Kundu, P.~Drineas, and M.~Magdon-Ismail.
\newblock Recovering {PCA} and sparse {PCA} via hybrid-(l1, l2) sparse sampling
  of data elements.
\newblock \emph{The Journal of Machine Learning Research}, 18\penalty0
  (1):\penalty0 2558--2591, 2017.

\bibitem[Lopes(2013)]{lopes2013estimating}
M.~Lopes.
\newblock Estimating unknown sparsity in compressed sensing.
\newblock In \emph{Proceedings of the 30th International Conference on Machine
  Learning}, volume~28, pages 217--225. PMLR, 2013.

\bibitem[Magen and Zouzias(2011)]{magen2011low}
A.~Magen and A.~Zouzias.
\newblock Low rank matrix-valued {C}hernoff bounds and approximate matrix
  multiplication.
\newblock In \emph{Proceedings of the 22nd {A}nnual ACM-SIAM {S}ymposium on
  {D}iscrete {A}lgorithms}, pages 1422--1436. SIAM, 2011.

\bibitem[Mroueh et~al.(2017)Mroueh, Marcheret, and Goel]{mroueh2017co}
Y.~Mroueh, E.~Marcheret, and V.~Goel.
\newblock {Co-Occurring Directions Sketching for Approximate Matrix Multiply}.
\newblock In \emph{Proceedings of the 20th International Conference on
  Artificial Intelligence and Statistics}, volume~54, pages 567--575. PMLR,
  2017.

\bibitem[Musco et~al.(2018)Musco, Netrapalli, Sidford, Ubaru, and
  Woodruff]{musco2017spectrum}
C.~Musco, P.~Netrapalli, A.~Sidford, S.~Ubaru, and D.~P. Woodruff.
\newblock Spectrum approximation beyond fast matrix multiplication: Algorithms
  and hardness.
\newblock In \emph{9th Innovations in Theoretical Computer Science Conference
  (ITCS 2018)}, volume~94 of \emph{Leibniz International Proceedings in
  Informatics (LIPICS)}, pages 8:1--8:21. Schloss Dagstuhl--Leibniz-Zentrum
  fuer Informatik, 2018.

\bibitem[Nguyen et~al.(2015)Nguyen, Drineas, and Tran]{nguyen2015tensor}
N.~H. Nguyen, P.~Drineas, and T.~D. Tran.
\newblock Tensor sparsification via a bound on the spectral norm of random
  tensors.
\newblock \emph{Information and Inference: A Journal of the IMA}, 4\penalty0
  (3):\penalty0 195--229, 2015.

\bibitem[Tropp(2012)]{tropp2012user}
J.~A. Tropp.
\newblock User-friendly tail bounds for sums of random matrices.
\newblock \emph{Foundations of Computational Mathematics}, 12\penalty0
  (4):\penalty0 389--434, 2012.

\bibitem[Ye et~al.(2016)Ye, Luo, and Zhang]{ye2016frequent}
Q.~Ye, L.~Luo, and Z.~Zhang.
\newblock Frequent direction algorithms for approximate matrix multiplication
  with applications in {CCA}.
\newblock In \emph{Proceedings of the Twenty-Fifth International Joint
  Conference on Artificial Intelligence, {IJCAI}}, pages 2301--2307.
  {IJCAI/AAAI} Press, 2016.

\end{thebibliography}

\end{document}